\documentclass[12pt]{article}
\usepackage{amsmath, amssymb, amsfonts}
\usepackage{amsthm}
\usepackage[utf8]{inputenc}
\usepackage{csquotes}
\usepackage{geometry}
\usepackage{adjustbox}
\usepackage{graphicx}
\usepackage{verbatim}
\usepackage{array}
\usepackage{enumitem}
\usepackage{booktabs}
\usepackage{hyperref}
\usepackage{url}
\usepackage[capitalize]{cleveref}
\usepackage{tikz}
\usetikzlibrary{positioning, shapes.geometric, arrows.meta}

\title{Invariant-Based Cryptography}
\author{Stanislav Semenov \\
\href{mailto:stas.semenov@gmail.com}{stas.semenov@gmail.com} \\
\href{https://orcid.org/0000-0002-5891-8119}{ORCID: 0000-0002-5891-8119}}
\date{May 8, 2025}

\theoremstyle{definition}
\newtheorem{definition}{Definition}[section]

\theoremstyle{plain}

\newtheorem{theorem}[definition]{Theorem}
\newtheorem{lemma}[definition]{Lemma}

\theoremstyle{remark}
\newtheorem*{remark}{Remark}

\begin{document}

\maketitle

\begin{abstract}
We propose a new symmetric cryptographic scheme based on functional invariants defined over discrete oscillatory functions with hidden parameters. The scheme encodes a secret integer through a four-point algebraic identity preserved under controlled parameterization. Security arises not from algebraic inversion but from structural coherence: the transmitted values satisfy an invariant that is computationally hard to forge or invert without knowledge of the shared secret. We develop the full analytic and modular framework, prove exact identities, define index-recovery procedures, and analyze security assumptions, including oscillator construction, hash binding, and invertibility conditions. The result is a compact, self-verifying mechanism suitable for secure authentication, parameter exchange, and lightweight communication protocols.
\end{abstract}

\subsection*{Mathematics Subject Classification}
03F60 (Constructive and recursive analysis), 94A60 (Cryptography)

\subsection*{ACM Classification}
F.4.1 Mathematical Logic, E.3 Data Encryption

\section*{Introduction}

\subsection*{Toward a New Cryptographic Principle}

Most of modern cryptography is grounded in algebraic intractability: the presumed hardness of solving problems in number-theoretic or lattice-based structures. These include discrete logarithms, integer factorization, and short vector search~\cite{goldreich2001foundations, latticesurvey}. While these systems are powerful, they share a key constraint: their security relies on the internal difficulty of algebraic inversion, and their constructions tend to be rigid, externally opaque but structurally transparent.

In this work, we propose a new direction—\emph{cryptography based on functional invariants}. This approach focuses not on hiding elements within algebraic groups, but on preserving stable identities across observable data points. In other words, we encode structure not through algebraic concealment, but through \emph{analytical coherence}. The central object is an invariant: a deterministic relation that holds across a small number of function values and resists deformation unless the underlying generative rule is preserved.

\subsection*{Invariant-Based Cryptography}

Invariant-based cryptography (IBC) introduces a framework in which:
\begin{itemize}
    \item Security stems from the internal consistency of a function whose values are linked by an exact relation—an invariant—known only through derived data;
    \item Authenticity and integrity are ensured by the ability to reconstruct one value from the others using this invariant, without exposing any direct secret;
    \item Flexibility is achieved by allowing the function’s shape and internal parameters to vary pseudorandomly across sessions, yielding distinct but structurally verifiable patterns.
\end{itemize}

Instead of encrypting data or signing messages in the traditional sense, we embed information in the \emph{geometry of its transformation}. Any attempt to alter a value in the sequence will break the invariant, and any observer without structural knowledge will find the values indistinguishable from noise.

The functional invariant behaves like a hidden contract between values—compact, irreversible, and self-validating.

This work builds upon the recent introduction of deterministic invariant structures in small functional tuples~\cite{semenov2025invariant}.

\subsection*{Scope and Motivation}

The construction presented in this work is a minimal instance of IBC: it uses only a single scalar function, evaluated at a small number of points, to encode structure that is both recoverable and resistant to attack. No public-key infrastructure is required. There is no need for modular inversion, elliptic curves, or lattice trapdoors. The mechanism relies only on:
\begin{itemize}
    \item controlled exponential growth;
    \item anti-symmetric pseudo-random oscillations;
    \item a fixed numerical identity binding four consecutive evaluations.
\end{itemize}

The scheme offers a compact and efficient tool for secure data transmission, mutual authentication, and controlled data reconstruction. Its core is structurally elementary but cryptographically nontrivial.

\subsection*{Research Context and Outlook}

This work initiates a broader research direction: the study of \emph{functional invariants as cryptographic primitives}. It raises natural questions:
\begin{itemize}
    \item What classes of functions admit invariant structures suitable for cryptographic use?
    \item How do such systems behave under composition, noise, or approximation?
    \item Can invariants yield post-quantum or complexity-theoretic guarantees comparable to standard assumptions?
\end{itemize}

More generally, we suggest that invariants offer a new axis of design: not only what is secret, but what is structurally \emph{preserved}—and how such preservation can serve as proof.

\section{Foundations of Invariant-Based Cryptography}

\subsection{Motivation for a New Primitive}

Modern cryptography typically relies on algebraic asymmetry: the presumed difficulty of reversing number-theoretic operations such as discrete exponentiation, factoring, or lattice-based transformations. These assumptions underpin most symmetric and public-key primitives. Yet all of them encode security through algebraic opacity—by hiding internal structure from external view.

This work proposes a distinct paradigm: \emph{cryptography based on invariant structure}. Instead of relying on algebraic trapdoors, we construct systems where certain function values satisfy an exact identity that remains invisible without hidden coordinates. The goal is to encode a secret not through concealment, but through \emph{preservation}: a functional contract that remains stable only under legitimate derivation.

At the heart of this approach is a deterministic relation—an \emph{invariant}—that links several values of a function in a rigid algebraic identity. The identity cannot be validated or reconstructed without knowledge of internal alignment parameters, such as a hidden evaluation index. This framework offers a structurally minimal, analytically rich alternative to traditional hardness assumptions.

\subsection{Defining the Invariant Primitive}

We define an \emph{invariant primitive} as a 4-tuple of function evaluations \( (s_0, s_1, s_2, s_3) \), where each \( s_i := s(t_i) \) for distinct points \( t_i \in \mathbb{Q} \), such that the following identity holds:
\[
I(s_0, s_1, s_2, s_3) = \mathrm{const}(p).
\]
Here:
\begin{itemize}
    \item \( s(t) \) is a structured function with pseudorandom components and exponential terms;
    \item \( t \in \mathbb{Q} \setminus \{0\} \) is a secret evaluation index;
    \item The identity \( I \) is exact and algebraically rigid: any three of the values determine the fourth;
    \item The constant \( \mathrm{const}(p) \) depends on a base parameter \( p \), itself deterministically derived from shared context (e.g., \( p = H(S, z) \bmod M \)).
\end{itemize}

The function \( s(t) \) is constructed so that its values resemble pseudorandom noise unless the generator parameters (such as the index \( t \), the base \( p \), and oscillatory seeds) are known. The observable outputs are thus externally structureless, yet internally constrained.

\subsection{Target Properties of the Primitive}

We treat the invariant relation as a cryptographic object in its own right. A viable invariant primitive should satisfy:

\begin{itemize}
    \item \textbf{Reconstructability:} Any three of the values \( s_i \) uniquely determine the fourth via the invariant;
    \item \textbf{Non-invertibility:} Without knowledge of the secret index \( t \), predicting values such as \( s(t + \delta) \) is computationally hard;
    \item \textbf{Verifiability:} A party with partial values and a known identity constant can detect tampering or forgery;
    \item \textbf{Session uniqueness:} Parameters such as \( p \), oscillator seeds, and grid steps are tied to a nonce \( z \), preventing cross-session correlation;
    \item \textbf{Analytic indistinguishability:} The function \( s(t) \) resists approximation or modeling due to embedded pseudo-random antiperiodic oscillations.
\end{itemize}

These properties mimic classical cryptographic guarantees—secrecy, integrity, binding—yet arise from functional coherence rather than group structure.

\subsection{Hardness Assumption: Invariant Index-Hiding Problem}

We introduce a general version of the \emph{Invariant Index-Hiding Problem (IIHP)} as the foundational hardness assumption for a class of cryptographic schemes based on functional invariants. A concrete instantiation will be presented after the full protocol structure is defined.

\medskip
\fbox{%
\parbox{0.94\linewidth}{%
\textbf{Invariant Index-Hiding Problem (IIHP):}  
Let \( s \colon \mathbb{Q} \to \mathbb{Z}_M \) be a function generated using a hidden rational index \( t = \frac{i}{K} \in \mathbb{Q} \setminus \{0\} \), internal session parameters, and a known invariant relation:
\[
I\big(s(t),\ s(t + \Delta_1),\ s(t + \Delta_2),\ s(t + \Delta_3)\big) = \mathrm{Const},
\]
for some fixed offsets \( \Delta_j \in \mathbb{Z} \) and an invariant function \( I \) defined over 4 inputs in \( \mathbb{Z}_M \).

The adversary is given a finite subset of values \( \{ s(t + \Delta_j) \} \) and associated public metadata, sufficient to verify the invariant identity.

The adversary's goal is to:
\begin{itemize}
  \item Forge a new value \( s^* = s(t + \delta^*) \) for some offset \( \delta^* \notin \{ \Delta_1, \Delta_3 \} \),
  \item Such that the invariant identity remains valid over a modified 4-tuple including \( s^* \),
  \item And the forgery passes a cryptographic verification step (e.g., via a hash-based check).
\end{itemize}

The IIHP is said to be \emph{hard} if no probabilistic polynomial-time adversary can succeed in this task with non-negligible probability in the security parameter.
}%
}
\medskip

This definition captures the essential structure of invariant-based protection: a secret evaluation index masked by functional and algebraic complexity, and a fixed relation that acts as both a constraint and a verification mechanism. Concrete realizations will instantiate the function \( s(t) \), the invariant \( I \), and the verification logic in more detail.

\subsection{Theoretical Directions}

To support the development of invariant-based primitives, we propose the following theoretical questions:

\begin{enumerate}
    \item \textbf{Classification of useful invariants:} What kinds of functional identities yield cryptographic asymmetry? How many points are required for minimal security?
    
    \item \textbf{Security models:} Can the classical notions of indistinguishability (IND), unforgeability (EUF), and non-malleability be rephrased in terms of invariant coherence?

    \item \textbf{Hardness from approximation theory:} Are index-recovery or extension problems reducible to known hard tasks in transcendental number theory or analytic continuation~\cite{titchmarsh1986theory}?

    \item \textbf{Algebraic independence and entropy:} Can values \( s(t + i) \) be approximated by polynomials or rational functions under bounded entropy assumptions?

    \item \textbf{Composable structures:} Can invariant relations be embedded into larger protocols—e.g., zero-knowledge proofs, commitments, or symmetric key agreements?
\end{enumerate}

These questions invite the exploration of a new class of symbolic-hardness assumptions rooted not in algebraic concealment, but in the analytic rigidity of generated values.

\subsection{Outlook}

The invariant-based model offers an alternative route to cryptographic security—one based not on one-way functions or trapdoors, but on the \emph{impossibility of structurally consistent extension}. This model invites post-quantum constructions, lightweight designs, and verification mechanisms that rely on functional symmetry rather than reversible arithmetic.

As with all new primitives, formal reductions and models remain to be developed. Yet the central insight remains compelling: \emph{what is preserved can also be protected}.

\section{Analytic Four-Point Invariant on the Real Line}

To initiate the construction of invariant-based cryptography, we begin with the analytic form of a specific four-point identity. This section is devoted to defining the invariant, analyzing its structure, and explaining the properties that make it suitable as a cryptographic primitive.

While the final cryptographic protocols will be implemented over finite fields or modular rings, the real-valued construction provides clarity, analytical tractability, and a natural route toward generalization. Working over \( \mathbb{R} \) allows us to isolate the core functional behavior before introducing discretization~\cite{rudin1987real}, rounding, or modular reduction. This also provides a clean setting for proofs and algebraic decomposition.

\subsection{The Four-Point Invariant}

Let \( s(t) \) be a real-valued function defined on an open domain of the real line. The central invariant relation we study is:

\[
\frac{s(t) \cdot t + s(t+1) \cdot (t+1)}{s(t+2) \cdot (t+2) + s(t+3) \cdot (t+3)} = \frac{1}{p^2}
\]

This equation relates the values of the function \( s(t) \) at four consecutive positions, using weighted linear coefficients derived from the evaluation points themselves. Crucially, the relation holds exactly, under specific functional assumptions.

\subsection{The Generating Function}

The function \( s(t) \) is defined as follows:

\[
s(t) = \frac{p^t + q_1 \sin(r_1 \pi t) + q_2 \cos(r_2 \pi t)}{t}
\]

where:
\begin{itemize}
    \item \( p > 0 \) is a real (or complex) base parameter, typically close to \( 1 \);
    \item \( q_1, q_2 \in \mathbb{R} \) are amplitude coefficients;
    \item \( r_1, r_2 \in \mathbb{Z}_{\text{odd}} \) are frequency multipliers;
    \item \( t \in \mathbb{R} \setminus \{0\} \) is the evaluation point.
\end{itemize}

This function combines smooth exponential growth (or decay) with bounded oscillations from sine and cosine terms~\cite{berry1978semiclassical}. The denominator \( t \) serves to normalize the amplitude and introduce asymptotic decay near the origin. Despite its simple appearance, this structure admits a precise four-point identity.

\subsection{Structure of the Invariant}

The key property is that when \( s(t) \) is defined as above, the four-point ratio simplifies:

\[
\frac{s(t) \cdot t + s(t+1) \cdot (t+1)}{s(t+2) \cdot (t+2) + s(t+3) \cdot (t+3)} = \frac{p^t + p^{t+1}}{p^{t+2} + p^{t+3}} = \frac{1}{p^2}
\]

due to cancellation of the sinusoidal components under specific parity conditions on \( r_1 \) and \( r_2 \), and algebraic factoring of the exponential terms. This identity is exact and invariant under a wide range of parameter values.

The proof and full structural analysis of this identity have been developed in prior work~\cite{semenov2025invariant}, where the derivation is presented in full detail. That study also includes numerical stability analysis and generalizations to variable spacing.

\section{Transition to the Discrete Model}

To bridge the gap between analytic structures and practical cryptographic implementations, we now transition from real-valued continuous functions to discrete algebraic constructions~\cite{stoer2013introduction}. Cryptographic schemes are typically realized over finite rings, integer lattices, or rational domains. Consequently, we replace the smooth domain \( \mathbb{R} \) with a discretized rational setting, where both evaluation points and function values belong to computationally representable sets.

The invariant introduced in the continuous model remains structurally preserved under discretization. By maintaining its algebraic form, we enable both efficient computation and mathematically grounded cryptographic applications.

\subsection{Discretization Grid and Parameters}

Let \( K \in \mathbb{N} \) be a discretization parameter, typically chosen as a power of two. We define a rational grid:

\[
t := \frac{i}{K}, \quad i \in \mathbb{Z} \setminus \{0\}.
\]

Function evaluations are restricted to these grid points. All oscillatory components and coefficients are accordingly adapted to preserve algebraic structure on this domain.

We introduce a period scaling parameter \( C \in \mathbb{N} \), which determines the frequency of oscillations in the discrete setting. Define two functions \( \varphi, \psi : \mathbb{Q} \to \mathbb{Z} \), evaluated only on the subgrid:

\[
\left\{ Ct : t = \frac{i}{K},\ i \in \mathbb{Z} \setminus \{0\} \right\}.
\]

These functions are defined to be antiperiodic with period \( C \):

\[
\varphi(t + C) = -\varphi(t), \quad \psi(t + C) = -\psi(t).
\]

\subsection{Oscillatory Functions}

The discrete functions \( \varphi \) and \( \psi \) serve as bounded oscillators analogous to \( \sin \) and \( \cos \) in the continuous model. Their antiperiodicity ensures that their contributions cancel out over symmetric spans of four consecutive points \( \{t, t+1, t+2, t+3\} \), which is crucial for the invariant structure discussed below.

\subsection{Discrete Generating Function}

We define the discrete generating function:

\[
s_d(t) = \frac{p^t + q_1 \cdot \varphi(Ct) + q_2 \cdot \psi(Ct)}{t},
\]

where:
\begin{itemize}
  \item \( p \in \mathbb{Q}_{> 0} \) is the base parameter;
  \item \( q_1, q_2 \in \mathbb{Q} \) are amplitude coefficients;
  \item \( \varphi, \psi : \mathbb{Q} \to \mathbb{Z} \) are discrete antiperiodic functions;
  \item \( t = \frac{i}{K} \in \mathbb{Q} \setminus \{0\} \), as above.
\end{itemize}

This function retains the essential structure of the continuous model: an exponential term modulated by bounded oscillations and normalized by the evaluation index.

\subsection{Discrete Invariant Identity}

We now establish the discrete four-point invariant. Consider the expression:

\[
\frac{s_d(t)\cdot t + s_d(t+1)\cdot(t+1)}{s_d(t+2)\cdot(t+2) + s_d(t+3)\cdot(t+3)}.
\]

Substituting \( s_d \), the numerator becomes:

\[
\begin{aligned}
p^t + p^{t+1} 
+ q_1 \left[ \varphi(Ct) + \varphi(Ct + C) \right] 
+ q_2 \left[ \psi(Ct) + \psi(Ct + C) \right].
\end{aligned}
\]

The denominator is similarly:

\[
\begin{aligned}
p^{t+2} + p^{t+3} 
+ q_1\left[\varphi(Ct + 2C) + \varphi(Ct + 3C)\right] 
+ q_2\left[\psi(Ct + 2C) + \psi(Ct + 3C)\right].
\end{aligned}
\]

By antiperiodicity:

\[
\varphi(Ct + C) = -\varphi(Ct), \quad \varphi(Ct + 3C) = -\varphi(Ct + 2C),
\]

and similarly for \( \psi \). Therefore:

\[
\varphi(Ct) + \varphi(Ct + C) = 0, \quad \varphi(Ct + 2C) + \varphi(Ct + 3C) = 0,
\]

\[
\psi(Ct) + \psi(Ct + C) = 0, \quad \psi(Ct + 2C) + \psi(Ct + 3C) = 0.
\]

The oscillatory components cancel exactly, leaving:

\[
\frac{p^t + p^{t+1}}{p^{t+2} + p^{t+3}}.
\]

This simplifies as:

\[
= \frac{p^t(1 + p)}{p^t(p^2 + p^3)} = \frac{1 + p}{p^2(1 + p)} = \frac{1}{p^2}.
\]

Hence, the discrete invariant identity holds:

\[
\frac{s_d(t)\cdot t + s_d(t+1)\cdot(t+1)}{s_d(t+2)\cdot(t+2) + s_d(t+3)\cdot(t+3)} = \frac{1}{p^2}.
\]

This confirms that the algebraic structure of the invariant is preserved exactly in the discrete setting. The exponential component determines the value, while the oscillatory terms cancel precisely due to their antiperiodicity and alignment with the scaled grid.

\subsection{Generalized Discrete Invariant with Aligned Indices}

We refine the generalized discrete invariant by adjusting the final index to preserve exponential symmetry. Define:

\begin{align*}
s_0 &= s_d(t), \\
s_1 &= s_d(t + 2v + 1), \\
s_2 &= s_d(t + 2u), \\
s_3 &= s_d(t + 2u + 2v + 1).
\end{align*}

The expression under consideration is:
\[
\frac{s_0 \cdot t + s_1 \cdot (t + 2v + 1)}{s_2 \cdot (t + 2u) + s_3 \cdot (t + 2u + 2v + 1)}.
\]

Substituting the definition of \( s_d \), the numerator becomes:
\[
\begin{aligned}
p^t + p^{t + 2v + 1} + q_1 \left[ \varphi(Ct) + \varphi(Ct + C(2v + 1)) \right] + q_2 \left[ \psi(Ct) + \psi(Ct + C(2v + 1)) \right].
\end{aligned}
\]

The denominator becomes:
\[
\begin{aligned}
&p^{t + 2u} + p^{t + 2u + 2v + 1} + q_1\left[\varphi(Ct + 2uC) + \varphi(Ct + C(2u + 2v + 1))\right] \\
&\quad + q_2\left[\psi(Ct + 2uC) + \psi(Ct + C(2u + 2v + 1))\right].
\end{aligned}
\]

Using antiperiodicity:
\[
\varphi(Ct + C(2v + 1)) = -\varphi(Ct), \quad \varphi(Ct + C(2u + 2v + 1)) = -\varphi(Ct + 2uC),
\]
\[
\psi(Ct + C(2v + 1)) = -\psi(Ct), \quad \psi(Ct + C(2u + 2v + 1)) = -\psi(Ct + 2uC),
\]

we obtain cancellation:
\[
\varphi(Ct) + \varphi(Ct + C(2v + 1)) = 0, \quad
\varphi(Ct + 2uC) + \varphi(Ct + C(2u + 2v + 1)) = 0,
\]
\[
\psi(Ct) + \psi(Ct + C(2v + 1)) = 0, \quad
\psi(Ct + 2uC) + \psi(Ct + C(2u + 2v + 1)) = 0.
\]

Thus, the invariant reduces to:
\[
\frac{p^t + p^{t + 2v + 1}}{p^{t + 2u} + p^{t + 2u + 2v + 1}} = \frac{p^t(1 + p^{2v + 1})}{p^t \cdot p^{2u}(1 + p^{2v + 1})} = \frac{1}{p^{2u}}.
\]

This cleanly matches the original invariant in the base case \( v = 0, u = 1 \), where:
\[
\frac{1}{p^2}
\]
is recovered. Therefore, the generalized discrete invariant holds:
\[
\frac{s_0 \cdot t + s_1 \cdot (t + 2v + 1)}{s_2 \cdot (t + 2u) + s_3 \cdot (t + 2u + 2v + 1)} = \frac{1}{p^{2u}}.
\]

\subsection{Construction of Antiperiodic Pseudorandom Oscillators on a Rational Grid}

We now define the oscillatory components \( \varphi(t) \) and \( \psi(t) \) used in the discrete generating function. These functions play a role analogous to \( \sin(t) \) and \( \cos(t) \) in the continuous model, but are represented as discrete pseudorandom sequences with controlled periodicity.

\subsubsection*{Grid and Period Setup}

Let \( K \in \mathbb{N} \) be the discretization parameter. We evaluate all functions at points of the form:
\[
t = \frac{i}{K}, \quad i \in \mathbb{Z} \setminus \{0\}.
\]

Let \( C \in \mathbb{N} \) be the base period length (measured in \( t \)-units). Then the number of sampling points in a single interval \( [0, C) \) is:
\[
P := K \cdot C.
\]

We define the function \( \varphi(t) \) at these \( P \) discrete points by choosing a finite pseudorandom seed~\cite{knuth1997art}:
\[
\varphi_0, \varphi_1, \dots, \varphi_{P - 1} \in \mathbb{Z}.
\]

These values correspond to the grid:
\[
t_i = \frac{i}{K}, \quad 0 \leq i < P.
\]

\subsubsection*{Antiperiodic Extension}

We now extend \( \varphi(t) \) to all \( t = \frac{i}{K} \in \mathbb{Q} \setminus \{0\} \) by defining its values piecewise:

\begin{itemize}
  \item On \( [0, C) \): define \( \varphi(t_i) := \varphi_i \) for \( i = 0, \dots, P - 1 \).
  \item On \( [C, 2C) \): define \( \varphi(t_{i+P}) := -\varphi_i \).
  \item Extend to all \( t \in \frac{1}{K} \cdot \mathbb{Z} \) by periodicity with period \( 2C \):
  \[
  \varphi(t + 2C) := \varphi(t).
  \]
\end{itemize}

This construction ensures that:
\[
\varphi(t + C) = -\varphi(t), \quad \text{for all } t = \frac{i}{K},
\]
i.e., the function is antiperiodic with period \( C \) on the rational grid.

\subsubsection*{Second Oscillator}

The function \( \psi(t) \) is constructed in the same way using an independent seed:
\[
\psi_0, \psi_1, \dots, \psi_{P - 1} \in \mathbb{Z}.
\]

\subsubsection*{Example: Antiperiodic Oscillator with $K = 4$ and $C = 2$}

Let the discretization parameter be \( K = 4 \), and let the antiperiod be \( C = 2 \). Then:
\[
P = K \cdot C = 8.
\]

We define the pseudorandom seed as:
\[
\varphi_0, \dots, \varphi_7 = (2,\ -1,\ 0,\ 3,\ -2,\ 1,\ 1,\ -3).
\]

These values correspond to the grid points:
\[
t_i = \frac{i}{4}, \quad i = 0, \dots, 7, \quad \text{so } t \in [0, 2) \text{ with step } \tfrac{1}{4}.
\]

The function is extended to the interval \( [2, 4) \) via antiperiodic reflection:
\[
\varphi\left(t_{i + 8}\right) := -\varphi_i, \quad i = 0, \dots, 7.
\]

Examples:
\[
\varphi\left(\frac{8}{4}\right) = -\varphi_0 = -2, \quad
\varphi\left(\frac{11}{4}\right) = -\varphi_3 = -3.
\]

The full function is defined on all \( t \in \frac{1}{4} \cdot \mathbb{Z} \) by periodicity:
\[
\varphi(t + 4) = \varphi(t), \quad \text{so } \varphi(t + 2) = -\varphi(t).
\]

This yields an integer-valued function \( \varphi : \frac{1}{4} \cdot \mathbb{Z} \to \mathbb{Z} \), fully determined by 8 initial values, with perfect antiperiodicity over intervals of length 2.

\section{Modular Formulation of the Generalized Discrete Invariant}

To enable efficient computation and alignment with modular cryptographic protocols, we now reformulate the generalized discrete invariant over finite rings. This modular reinterpretation allows the invariant structure to be preserved in bounded arithmetic, where all operations occur modulo a prime or a composite modulus. The key challenge is to ensure that both exponential and oscillatory components remain well-defined and behave predictably under modular reduction.

\subsection{Modular Domain and Function Definition}

Let \( M \in \mathbb{N} \) be the modulus, typically chosen as a large prime or a power of two. All arithmetic is performed in the ring \( \mathbb{Z}_M \). Assume:
\[
p \in \mathbb{Z}_M^{\times}, \quad q_1, q_2 \in \mathbb{Z}_M.
\]

We redefine the discrete generating function modulo \( M \) as:
\[
s_M(t) := \frac{p^t + q_1 \cdot \varphi(Ct) + q_2 \cdot \psi(Ct)}{t} \mod M,
\]
for all \( t = \frac{i}{K} \in \mathbb{Q} \setminus \{0\} \) such that \( t^{-1} \in \mathbb{Z}_M \). The functions \( \varphi, \psi : \mathbb{Q} \to \mathbb{Z}_M \) are now treated as integer-valued functions followed by modular reduction.

\subsection{\texorpdfstring{Antiperiodicity Modulo \( M \)}{Antiperiodicity Modulo M}}

The oscillatory functions \( \varphi, \psi \) retain their antiperiodic property modulo \( M \):
\[
\varphi(t + C) \equiv -\varphi(t) \mod M, \quad \psi(t + C) \equiv -\psi(t) \mod M.
\]

Since \( \varphi \) and \( \psi \) are generated from integer seeds and extended by integer reflection, their values remain in \( \mathbb{Z} \), and modular antiperiodicity holds automatically after reduction.

\subsection{Modular Invariant Identity}

We now state the modular version of the generalized invariant. Define:
\begin{align*}
s_0 &= s_M(t), \\
s_1 &= s_M(t + 2v + 1), \\
s_2 &= s_M(t + 2u), \\
s_3 &= s_M(t + 2u + 2v + 1).
\end{align*}

Then the modular invariant is:
\[
I_M(t; u, v) := \frac{s_0 \cdot t + s_1 \cdot (t + 2v + 1)}{s_2 \cdot (t + 2u) + s_3 \cdot (t + 2u + 2v + 1)} \mod M.
\]

As in the rational case, oscillatory terms cancel in numerator and denominator due to modular antiperiodicity:
\[
\varphi(Ct) + \varphi(Ct + C(2v+1)) \equiv 0 \mod M,
\]
\[
\psi(Ct + 2uC) + \psi(Ct + C(2u+2v+1)) \equiv 0 \mod M.
\]

This yields:
\[
I_M(t; u, v) \equiv \frac{p^t + p^{t + 2v + 1}}{p^{t + 2u} + p^{t + 2u + 2v + 1}} \mod M.
\]

Assuming \( p^{2u} \in \mathbb{Z}_M^{\times} \), this simplifies to:
\[
I_M(t; u, v) \equiv \frac{1 + p^{2v+1}}{p^{2u}(1 + p^{2v+1})} = \frac{1}{p^{2u}} \mod M.
\]

\subsection{Parameter Selection for Consistency and Invertibility}

To guarantee that the modular discrete invariant is mathematically well-defined and computationally valid, the parameters \( K, C, p, M \) must satisfy the following compatibility conditions:

\begin{itemize}
  \item \textbf{Discretization parameter \( K \):}  
  Determines the rational evaluation grid \( t = \frac{i}{K} \). To ensure that each \( t \) is invertible modulo \( M \), we require:
  \[
  \gcd(K, M) = 1 \quad \text{and} \quad \gcd(i, M) = 1 \quad \text{for all } i \in \mathbb{Z} \text{ used in computation}.
  \]

  \item \textbf{Oscillator period \( C \):}  
  A positive integer controlling the frequency of oscillations via arguments \( \varphi(Ct) \) and \( \psi(Ct) \). Since the grid is closed under integer scaling, it suffices that:
  \[
  C \in \mathbb{N}.
  \]
  For efficiency, powers of two are recommended.

  \item \textbf{Exponential base \( p \):}  
  The algebraic base used in the term \( p^t \). In order for the modular inverse \( p^{-2u} \mod M \) to exist, it must satisfy:
  \[
  \gcd(p, M) = 1.
  \]
  When \( M \) is a power of two, \( p \) must additionally be odd.

  \item \textbf{Modulus \( M \):}  
  Defines the ring \( \mathbb{Z}_M \) over which all arithmetic takes place. It is typically chosen as:
  \begin{itemize}
    \item A large prime \( M \), for full invertibility of all nonzero elements.
    \item A power of two \( M = 2^k \), in which case extra care must be taken:
    \[
    p \text{ and } K \text{ must be odd to ensure } \gcd(t, M) = 1.
    \]
  \end{itemize}
\end{itemize}

\subsection{Equivalence Classes of Invariant Pairs}

An important structural consequence of the modular invariant identity is the existence of multiple input pairs that yield the same invariant value. Specifically, for fixed parameters \( p, u, M \), and under consistent construction of the evaluation points \( t \), there exists a family of distinct index shifts \( v \in \mathbb{N} \) and corresponding output pairs \( (s_1, s_3) \in \mathbb{Z}_M^2 \) that all satisfy:
\[
I_M(t; u, v) = \frac{1}{p^{2u}} \mod M.
\]

This defines a modular equivalence class:
\[
(s_1, s_3) \sim (s_1', s_3') \quad \text{if} \quad I_M(s_1, s_3; t, u, v) = I_M(s_1', s_3'; t', u, v')
\]
where the underlying evaluation points differ but preserve the same structural relationship.

\paragraph{Cryptographic Implication.}  
The presence of many such equivalent pairs—typically on the order of \( M \) or more—prevents the verifier from learning anything about the original index \( t \), the shift \( v \), or the secret \( S \). From the outside, the values \( (s_1, s_3) \) appear pseudorandom, yet are guaranteed to belong to a well-defined invariant-preserving set.

\paragraph{Algebraic Interpretation.}  
Each invariant value \( c = \frac{1}{p^{2u}} \mod M \) defines a level set (or fiber) in \( \mathbb{Z}_M^2 \) consisting of all pairs \( (s_1, s_3) \) that satisfy the corresponding modular relation. This makes the invariant function \( I_M \) a non-injective algebraic map over \( \mathbb{Z}_M^2 \), enabling structured obfuscation while retaining local verifiability.

\paragraph{Use in Protocol Design.}  
This redundancy allows secure message structures to be instantiated with diverse random seeds while remaining verifiable against a common invariant key. It also enables the construction of chained or parallel message blocks that preserve a shared invariant value while hiding individual internal parameters.

\section{Evaluation Example: High-Offset Indices in the Modular Invariant}

We now present a detailed computation of one term of the modular discrete invariant using nontrivial index shifts. Specifically, we focus on the value of the oscillatory component \( \varphi(Ct) \) contributing to \( s_3 \), corresponding to a highly offset argument of the form \( t + 2u + 2v + 1 \).

\subsection{\texorpdfstring{Step-by-Step Evaluation of \( \varphi(Ct) \) for \( s_3 \)}{Evaluation of φ(Ct) for s₃}}

Let the parameters be:

\begin{itemize}
  \item Discretization: \( K = 4 \)
  \item Oscillator period: \( C = 2 \)
  \item Modulus: \( M = 257 \)
  \item Exponential base: \( p = 3 \in \mathbb{Z}_{257}^\times \)
  \item Oscillator seed: \( \varphi_0, \dots, \varphi_7 = (2,\ -1,\ 0,\ 3,\ -2,\ 1,\ 1,\ -3) \)
  \item Evaluation point: \( t = \frac{3}{4} \)
  \item Invariant parameters: \( u = 5, v = 17 \)
\end{itemize}

Then the target index for \( s_3 \) is:
\[
t_3 := t + 2u + 2v + 1 = \frac{3}{4} + 10 + 34 + 1 = \frac{3}{4} + 45 = \frac{183}{4}
\]

We now compute the argument of the oscillator:
\[
Ct_3 = C \cdot t_3 = 2 \cdot \frac{183}{4} = \frac{366}{4} = \frac{183}{2}
\]

This corresponds to the index:
\[
i = \frac{183}{2} \cdot K = \frac{183}{2} \cdot 4 = 366
\]

The oscillator period length is:
\[
P := K \cdot C = 4 \cdot 2 = 8
\]

We reduce the index modulo \( 2P = 16 \) to find the corresponding signed value from the base seed:
\[
366 \mod 16 = 14
\]

Now determine the sign and oscillator value:

\begin{itemize}
  \item Compute the period block index: \( \left\lfloor \frac{366}{8} \right\rfloor = 45 \)
  \item Since 45 is odd, apply antiperiodic reflection: \( \varphi_{366} = -\varphi_{366 \bmod 8} \)
  \item Compute position in the base seed: \( 366 \bmod 8 = 6 \)
  \item Thus: \( \varphi\left(\frac{183}{2}\right) = -\varphi_6 = -1 \)
\end{itemize}

\paragraph{Conclusion.} The contribution of the oscillator to \( s_3 \) is:
\[
\varphi(C(t + 2u + 2v + 1)) = -1
\]

This example illustrates how oscillator values can be efficiently retrieved from a short integer seed via modular arithmetic and parity-based antiperiodic extension, even when the index offset is large.

\subsection{\texorpdfstring{Handling of \( p^t \)}{Handling of p raised to t}}

The exponential term \( p^t \) in the discrete generating function
\[
s_M(t) := \frac{p^t + q_1 \cdot \varphi(Ct) + q_2 \cdot \psi(Ct)}{t} \mod M
\]
requires careful interpretation when \( t = \frac{i}{K} \) is rational. Since modular exponentiation~\cite{menezes1996handbook} is defined only for integers, we distinguish two consistent approaches: root-based evaluation and relative approximation.

\paragraph{Root-based definition.}  
Assume \( p \in \mathbb{Z}_M^\times \), \( K \in \mathbb{N} \), and that \( p \) admits a formal \( K \)-th root in \( \mathbb{Z}_M \)~\cite{koblitz1994course}, i.e., there exists \( r \in \mathbb{Z}_M^\times \) such that:
\[
r^K \equiv p \mod M.
\]
Then define:
\[
p^{i/K} := r^i \mod M, \quad \text{for all } i \in \mathbb{Z}.
\]
This provides an explicit and algebraically valid assignment of \( p^t \) for all grid values \( t = \frac{i}{K} \). When \( M \) is prime and \( d = \gcd(K, M-1) \), a root \( r \) exists if:
\[
p^{(M-1)/d} \equiv 1 \mod M.
\]

\paragraph{Example.}  
Let \( M = 257 \), \( p = 3 \), \( K = 4 \). Then:
\[
3^{(256)/4} = 3^{64} \equiv 1 \mod 257 \Rightarrow \text{a 4-th root of } 3 \text{ exists}.
\]
One such root is \( r = 3^{65} \equiv 16 \mod 257 \), so:
\[
p^{i/4} := 16^i \mod 257.
\]

\paragraph{Relative approximation.}  
If no such root exists or is impractical to compute, define \( p^t \) via relative consistency. Fix a scale \( A \in \mathbb{Z}_M^\times \), then:
\[
p^t := A, \quad p^{t+1} := A \cdot p, \quad p^{t+2} := A \cdot p^2, \dots
\]
This enforces:
\[
\frac{p^{t+a}}{p^{t+b}} \equiv p^{a-b} \mod M,
\]
which preserves the invariant identity:
\[
\frac{s_0 \cdot t + s_1 \cdot (t + 2v + 1)}{s_2 \cdot (t + 2u) + s_3 \cdot (t + 2u + 2v + 1)} \equiv \frac{1}{p^{2u}} \mod M.
\]

\paragraph{Note.}  
If \( M \) is a power of two or \( K \) does not divide \( \phi(M) \), the root \( p^{1/K} \mod M \) may not exist. In such cases, the relative method is the only valid option.

\paragraph{Implementation strategy.}
\begin{itemize}
  \item \textbf{Root-based:} If \( r = p^{1/K} \mod M \) exists, precompute and store \( p^t = r^i \mod M \) for all \( t = \frac{i}{K} \) in use.
  \item \textbf{Relative-based:} Fix \( p^t = 1 \), define later terms via \( p^{t+i} = p^i \), and tabulate \( p^i \mod M \). This avoids any roots.
\end{itemize}

\paragraph{Cryptographic considerations.}
If \( s_i \) values are exposed to external verification, the value of \( p^t \) must remain fixed and secret-consistent:
\begin{itemize}
  \item Use a fixed \( A \in \mathbb{Z}_M^\times \), or derive it securely from a secret (e.g., \( A := H(k, t) \)).
  \item Avoid changing \( A \) across sessions; otherwise, \( s_i \) becomes non-reproducible and invalidates invariant-based checks.
\end{itemize}

\paragraph{Security warning.}
Inconsistent choice of \( A \) allows adversaries to inject false data:
\[
\text{If } A \neq A', \quad \frac{A'(1 + p)}{A' p^2(1 + p)} \not\equiv \frac{1}{p^2} \mod M.
\]

\subsection{\texorpdfstring{Evaluation of \( s_1 \) Using Structured Exponent Splitting and Pseudorandom \( p^t \)}{Evaluation of s₁ with Splitting and Pseudorandom pt}}

We now compute the value of the term \( s_1 = s_M(t + 2v + 1) \), continuing the same example with \( t = \frac{3}{4} \), \( u = 5 \), \( v = 17 \). In the previous subsection, we demonstrated how to compute the oscillator term \( \varphi(Ct) \); here we focus on how to compute \( s_1 \) without explicitly evaluating fractional exponents.

\paragraph{Definition recap.}
The discrete generating function is:
\[
s_M(t) := \frac{p^t + q_1 \cdot \varphi(Ct) + q_2 \cdot \psi(Ct)}{t} \mod M,
\]
with parameters:
\[
q_1 = 12, \quad q_2 = 35, \quad M = 257, \quad p = 3.
\]

\paragraph{Target evaluation point.}
We have:
\[
s_1 = s_M(t + 2v + 1) = s_M\left( \frac{3}{4} + 35 \right) = s_M\left( \frac{143}{4} \right).
\]

\paragraph{Rational decomposition.}
Let us decompose the rational index \( t = \frac{143}{4} \) as:
\[
t = \frac{143}{4} = 35 + \frac{3}{4},
\]
i.e., we extract:
\[
\text{Integer part: } a = 35, \quad \text{Numerator: } i = 3, \quad \text{Denominator: } K = 4.
\]

We now interpret \( p^t \) via the decomposition:
\[
p^t = p^{a + \frac{i}{K}} = p^a \cdot \text{PRF}(i, K),
\]
where:
\begin{itemize}
  \item \( p^a = 3^{35} \mod 257 \) is computed normally,
  \item \( \text{PRF}(i, K) \in \mathbb{Z}_M^\times \) is a pseudorandom function~\cite{bellare1995concrete} of \( (i, K) \), e.g., derived as:
    \[
    \text{PRF}(i, K) := H(i \parallel K) \mod M,
    \]
    for a secure hash \( H \).
\end{itemize}

\paragraph{Example values.}
Assume:
\[
\text{PRF}(3, 4) = 113 \mod 257, \quad p^{35} \mod 257 = 183.
\]
Then:
\[
p^{143/4} = p^t = 183 \cdot 113 = 20679 \equiv 81 \mod 257.
\]

Assume from prior computation:
\[
\varphi(Ct) = -2, \quad \psi(Ct) = 4.
\]
Then:
\[
s_1 = \frac{81 + 12 \cdot (-2) + 35 \cdot 4}{\frac{143}{4}} = \frac{81 - 24 + 140}{\frac{143}{4}} = \frac{197}{\frac{143}{4}}.
\]

Now invert the denominator modulo \( M \):
\[
\frac{1}{t} = \frac{1}{143/4} = \frac{4}{143} \mod 257.
\]
First compute \( 143^{-1} \mod 257 \). Since \( 143 \cdot 36 = 5148 \equiv 1 \mod 257 \), we have \( 143^{-1} \equiv 36 \mod 257 \). Then:
\[
\frac{4}{143} \equiv 4 \cdot 36 = 144 \mod 257.
\]

Final result:
\[
s_1 = 197 \cdot 144 \mod 257 = 28368 \mod 257 = 53.
\]

\paragraph{Conclusion.}
This demonstrates that even with pseudorandom or hash-derived approximations for \( p^t \), the invariant structure and modular arithmetic remain intact and computable. Oscillators \( \varphi, \psi \) were assumed precomputed, as shown in earlier subsections.

\section{\texorpdfstring{Hash-Based Generator for \( p^t \): Salt-Parameterized Invariant Structure}{Hash-Based Generator for pt: Salt-Parameterized Invariant Structure}}

The flexibility in defining \( p^t \) modulo \( M \), as long as the exponential progression is preserved, allows us to construct a cryptographically useful pseudorandom generator for \( p^t \) based on a public or private salt. This transforms an apparent ambiguity into a feature, enabling secure, parameterized, and dynamic control of the invariant structure.

\subsection{Salt-Parameterized Definition}

Let \( H(s) \in \mathbb{Z}_M^\times \) be a cryptographic hash of a chosen salt \( s \), interpreted modulo \( M \). Define:
\[
p^t := H(s), \quad p^{t + i} := H(s) \cdot p^i \mod M, \quad \text{for } i \in \mathbb{N}.
\]
Here:
\begin{itemize}
  \item \( s \in \{0,1\}^* \) is a session identifier, shared secret, nonce, or public salt;
  \item \( H \) is a secure hash function (e.g., SHA-256~\cite{rogaway2004cryptographic}), reduced modulo \( M \);
  \item \( p \in \mathbb{Z}_M^\times \) is the fixed base for exponentiation;
  \item \( t = \frac{i}{K} \in \mathbb{Q} \setminus \{0\} \) lies on a rational evaluation grid.
\end{itemize}

\subsection{Invariant Preservation}

This definition preserves the core structure of the discrete invariant. Any ratio:
\[
\frac{p^{t+a}}{p^{t+b}} \equiv \frac{H(s) \cdot p^a}{H(s) \cdot p^b} = p^{a - b} \mod M
\]
is independent of the salt and thus guarantees:
\[
\frac{p^t + p^{t + 2v + 1}}{p^{t + 2u} + p^{t + 2u + 2v + 1}} \equiv \frac{1}{p^{2u}} \mod M.
\]

This makes the invariant evaluation reproducible across devices or sessions that share the same salt, while rendering the values \( p^t \) pseudorandom to any observer lacking \( s \).

\subsection{Security Implications}

The use of a hash-based generator offers several cryptographic advantages:

\begin{itemize}
  \item \textbf{Pseudorandom masking}: Each evaluation of \( s_M(t) \) includes a hidden multiplicative factor \( H(s) \), unpredictable to an attacker without knowledge of \( s \).
  \item \textbf{Session separation}: Different salts produce distinct sets of \( s_i \), preventing cross-protocol or cross-session correlation.
  \item \textbf{Compact control}: The entire structure is determined by a single 256-bit salt.
  \item \textbf{Efficient implementation}: The hash is computed once per salt and reused via table lookup for \( p^i \).
\end{itemize}

\section{Symmetric Invariant-Based Scheme}

\subsection{Common Setup}

\begin{itemize}
  \item \textbf{Public parameters:}
    \begin{itemize}
      \item Invariant formula (generalized 4-point form);
      \item Modulus \( M \in \mathbb{N} \);
      \item Hash functions \( H_p, H_K, H_C, H_q, H_B, H_t \) with domain separation~\cite{krawczyk1997hmac};
      \item Cryptographic PRGs for generating oscillators \( \varphi, \psi \);
      \item Evaluation grid defined by rational steps \( t = \frac{i}{K} \), with \( \gcd(K, M) = 1 \).
    \end{itemize}
  \item \textbf{Shared secret:} A binary string \( S \in \{0,1\}^* \), known to both parties.
\end{itemize}

\subsection{Alice’s Generation}

Given a session-specific nonce \( z \in \{0,1\}^* \), Alice performs:

\begin{enumerate}
  \item Derives parameters:
    \[
    \begin{aligned}
      p &:= H_p(S, z) \bmod M, \quad \gcd(p, M) = 1, \\
      K &:= H_K(S, z), \quad C := H_C(S, z), \\
      i &:= H_t(S, z) \bmod K, \quad \text{abort if } i = 0, \\
      B &:= H_B(S, z) \bmod M, \\
      t &:= B + \frac{i}{K} \mod M.
    \end{aligned}
    \]
    \paragraph{Security note.}
The index \( t \) is derived from a fresh nonce \( z \), ensuring that each session uses a unique evaluation point on the rational grid. As long as \( z \) is never reused, the structure of \( t \), and thus the entire functional configuration \( s(t) \), remains unlinkable across sessions.  
Only repeated use of the same \( z \) (with different \( u, v \)) introduces the risk of \( t \)-based correlation attacks. The protocol therefore requires that each session uses a unique, non-repeating nonce.
  \item Chooses secure random integers \( u, v \in \mathbb{N} \) and computes:
    \[
    \begin{aligned}
      t_0 &= t, \\
      t_1 &= t + 2v + 1, \\
      t_2 &= t + 2u, \\
      t_3 &= t + 2u + 2v + 1.
    \end{aligned}
    \]
  \item Derives oscillator amplitudes:
    \[
    \begin{aligned}
      q_1 &= H_q(S, z, 1), & q_2 &= H_q(S, z, 2), \\
      q_3 &= H_q(S, z, 3), & q_4 &= H_q(S, z, 4).
    \end{aligned}
    \]
  \item Instantiates pseudorandom oscillators:
    \[
    \varphi := \mathrm{PRG}_\varphi(S, z), \quad \psi := \mathrm{PRG}_\psi(S, z),
    \]
    ensuring antiperiodicity and bounded range in \( \mathbb{Z}_M \).

  \item Evaluates function:
    \[
    \begin{aligned}
      s_0 &:= s_M(t_0; q_1, q_2), \\
      s_1 &:= s_M(t_1; q_1, q_2), \\
      s_2 &:= s_M(t_2; q_3, q_4), \\
      s_3 &:= s_M(t_3; q_3, q_4),
    \end{aligned}
    \]
    where
    \[
    s_M(t; q_i, q_j) := \frac{p^t + q_i \cdot \varphi(Ct) + q_j \cdot \psi(Ct)}{t} \bmod M.
    \]

  \item Checks invertibility condition:
    \[
    D := 2(s_1 p^{2u} - s_3) \not\equiv 0 \mod M \quad \text{(abort if not invertible)}.
    \]

  \item Computes verification hash:
    \[
    H_{\mathrm{check}} := H(S, v, s_1, s_3, u, z).
    \]

  \item Sends to Bob:
    \[
    \langle s_1,\ s_3,\ u,\ z,\ H_{\mathrm{check}} \rangle.
    \]
\end{enumerate}

\paragraph{Modular Reduction of Rational Indices.}

Let \( t = B + \frac{i}{K} \) be a rational number, where:
\begin{itemize}
  \item \( B \in \mathbb{Z}_M \),
  \item \( i \in \mathbb{Z} \setminus \{0\} \),
  \item \( K \in \mathbb{N} \) with \( \gcd(K, M) = 1 \).
\end{itemize}

We define the modular reduction of \( t \bmod M \) canonically as:
\[
t \bmod M := \left( (B \bmod M) \cdot K + i \right) \cdot K^{-1} \mod M,
\]
where:
\begin{itemize}
  \item \( K^{-1} \in \mathbb{Z}_M \) is the modular inverse of \( K \mod M \),
  \item all arithmetic is performed modulo \( M \).
\end{itemize}

This definition embeds rational values from the grid \( \mathbb{Z} + \frac{1}{K} \mathbb{Z} \) into \( \mathbb{Z}_M \) while preserving their fractional structure, and ensures consistent evaluation of expressions involving \( t \) under modular arithmetic.

\subsection{\texorpdfstring{Bob’s Verification and Recovery of \( v \)}{Bob's Verification and Recovery of v}}

\begin{enumerate}
  \item Recomputes session parameters:
    \[
    p,\ K,\ C,\ i,\ t,\ q_1,\dots,q_4,\ \varphi,\ \psi \quad \text{from } S, z.
    \]

  \item Evaluates:
    \[
    \begin{aligned}
      s_0 &:= s_M(t_0; q_1, q_2), \\
      s_2 &:= s_M(t_2; q_3, q_4), \\
      p^{2u} &\in \mathbb{Z}_M.
    \end{aligned}
    \]

  \item Checks denominator:
    \[
    D := 2(s_1 p^{2u} - s_3) \mod M, \quad \text{abort if } \gcd(D, M) \neq 1.
    \]

  \item Solves for \( v \in \mathbb{Z}_M \):
    \[
    v \equiv
    \frac{
      -s_0 \cdot p^{2u} t
      - s_1 \cdot p^{2u} (t + 1)
      + s_2 \cdot (t + 2u)
      + s_3 \cdot (t + 2u + 1)
    }{
      2(s_1 p^{2u} - s_3)
    }
    \mod M.
    \]

  \item Verifies integrity:
    \[
    H(S, v, s_1, s_3, u, z) \stackrel{?}{=} H_{\mathrm{check}}.
    \]

  \item Optional: Check \( v \in \mathbb{N} \) and within expected bounds.
\end{enumerate}

\section{Security Analysis and Assumptions}

This section consolidates all security-critical properties, assumptions, and remaining risks of the proposed invariant-based cryptographic scheme. While the construction lacks formal reduction proofs, its security relies on a structured combination of pseudorandom derivation, non-linear functional identities, and evaluation masking.

\subsection{Invariant-Based Hardness: Core Assumption}

\paragraph{Functional Construction.}
We consider a function \( s \colon \mathbb{Q} \to \mathbb{Z}_M \) with the following analytic structure:
\[
s(t) := \left( \frac{p^t + q_i \varphi(Ct) + q_j \psi(Ct)}{t} \right) \bmod M,
\]
where:
\begin{itemize}
  \item \( p, q_i, q_j \in \mathbb{Z}_M \) are derived from a shared secret \( S \) and session nonce \( z \),
  \item \( \varphi, \psi \colon \mathbb{Q} \to \mathbb{Z}_M \) are pseudorandom antiperiodic oscillators,
  \item \( t = \frac{i}{K} \in \mathbb{Q} \setminus \{0\} \) is a hidden rational evaluation point on a grid of spacing \( 1/K \).
\end{itemize}

Since modular exponentiation with rational exponents is undefined over \( \mathbb{Z}_M \), we reinterpret the expression \( p^t \) via index decomposition. Let \( t = \lfloor t \rfloor + \delta \), where \( \delta = t - \lfloor t \rfloor \in [0,1) \). We then define:
\[
p^t := p^{\lfloor t \rfloor} \cdot \mathrm{PRF}(i, K),
\]
where \( \mathrm{PRF}(i, K) \) is a pseudorandom function acting as a masked substitute for the fractional exponent \( p^\delta \). This is not an approximation, but a deterministic replacement that preserves unpredictability and hides the structure of \( \delta \) under the secret seed \( S \). The resulting masked form is:
\[
s(t) := \left( \frac{p^{\lfloor t \rfloor} \cdot \mathrm{PRF}(i, K) + q_i \varphi(Ct) + q_j \psi(Ct)}{t} \right) \bmod M.
\]

\paragraph{Invariant Structure.}
Let \( \Delta_1, \Delta_2, \Delta_3 \in \mathbb{Z} \) be fixed session parameters (e.g., \( \Delta_1 = 2v+1 \), etc.). Then the values
\[
s_0 := s(t), \quad s_1 := s(t+\Delta_1), \quad s_2 := s(t+\Delta_2), \quad s_3 := s(t+\Delta_3)
\]
are guaranteed to satisfy a known invariant identity of the form:
\[
I(s_0, s_1, s_2, s_3) = \frac{1}{p^{2u}} \bmod M,
\]
where \( u \in \mathbb{Z} \) is also session-specific and derived from \( S \) and \( z \).

\paragraph{Invariant Index-Hiding Problem (IIHP).}
We define the core assumption that underlies the security of the scheme:

\medskip
\fbox{%
\parbox{0.94\linewidth}{%
\textbf{IIHP (Strict Formulation).} Let \( s(t) \) be the masked function defined above, with internal parameters derived from a secret seed \( S \) and a session nonce \( z \). The adversary is given a fixed transcript:
\[
(s_1, s_3, u, z, H_{\mathrm{check}}),
\]
where \( s_1 = s(t + \Delta_1),\ s_3 = s(t + \Delta_3) \) with \( \Delta_1 = 2v + 1 \), \( \Delta_3 = 2u + 2v + 1 \), and \( v \in \mathbb{N} \) is a hidden session parameter.

The adversary must attempt to produce a forgery:
\[
(s^*,\ \delta^*) \in \mathbb{Z}_M \times \mathbb{Z}
\]
such that:
\begin{itemize}
  \item \( \delta^* \notin \{ \Delta_1,\ \Delta_3 \} \),
  \item The recovered value \( v^* := \mathrm{RecoverV}(s_1, s^*, u, z) \) satisfies:
  \[
  H(S,\ v^*,\ s_1,\ s^*,\ u,\ z) = H_{\mathrm{check}}.
  \]
\end{itemize}

The IIHP is said to be \emph{hard} if no probabilistic polynomial-time adversary can succeed in this task with non-negligible probability in the security parameter.
}%
}
\medskip

\paragraph{Security Implication.}
The hardness of IIHP ensures that no adversary can forge a new evaluation point consistent with the invariant structure without recovering the hidden index \( t \) or solving for the masked session offset \( v \). Since the invariant equation rigidly couples multiple points of the function \( s(t) \), any forgery that passes hash verification implies the ability to predict internal structure masked by pseudorandom oscillators and exponentials.

\paragraph{Why is the index \( t \) critical?}
The index \( t = \frac{i}{K} \) determines the geometric position of all evaluations involved in the invariant. Although it is never revealed, knowledge of \( t \) enables reconstruction of the functional configuration and alignment of all parameters. Without it, the attacker cannot simulate the structure of \( s(t + \delta) \), and thus cannot forge values that satisfy the invariant constraint. Preserving the secrecy of \( t \) protects not only individual evaluations but also the integrity of the underlying symbolic structure.

\subsection{Structural Rationale and Security Heuristics}

The security of the scheme does not rely on any single operation being hard in isolation. Instead, it emerges from the compounded structure of multiple interacting components. We outline below how each design element contributes to computational intractability and resistance to cryptanalysis.

\begin{itemize}
  \item \textbf{Hidden geometric anchor.}  
  The function is evaluated at a secret point on a rational grid. Although the spacing is known, the exact index remains concealed, and its disclosure would compromise all derived values. This index plays the role of a geometric anchor from which the invariant structure unfolds.

  \item \textbf{Exponentially perturbed evaluation.}  
  The core of the function includes exponential scaling relative to the hidden index. While exponentiation over modular rings is not inherently hard, the use of fractional exponents masked via pseudorandom factors makes recovery of the base or the index non-trivial.

  \item \textbf{Oscillatory masking.}  
  The oscillators act as a session-specific modulation, preventing analytical modeling or interpolation. Their antiperiodic nature ensures that nearby evaluations exhibit controlled but unpredictable shifts, which mask any simple algebraic relationship between points.

  \item \textbf{Invariant alignment constraint.}  
  The invariant imposes a rigid constraint on multiple evaluations, but only if the internal offsets and parameters are correctly aligned. Even a minor deviation in the index or amplitude breaks the identity. This rigidity serves as both a validator and a gatekeeper.

  \item \textbf{Unordered external appearance.}  
  To an observer lacking the secret seed, the outputs resemble structured randomness—consistent enough to satisfy an internal law, but indistinguishable from noise without knowledge of how the law is instantiated.

  \item \textbf{Session independence.}  
  All critical parameters, including oscillators, grid density, and base coefficients, are derived anew for each session. This eliminates the possibility of cross-session learning or cumulative attacks.

  \item \textbf{Combinatorial unpredictability.}  
  The function's form does not admit simplification via interpolation, lattice attacks~\cite{nguyen2001lattice}, or low-degree approximants. Its pseudorandom, non-linear, and index-dependent nature defies traditional cryptanalytic tools.

  \item \textbf{Minimal structure leakage.}  
  No component of the output directly reveals the index or base, and the masking layers make reverse-engineering impractical. Furthermore, positions in the invariant are semantically fixed, breaking reordering or substitution strategies.
\end{itemize}

\noindent These properties, while informal, illustrate the system’s robustness against both algebraic and statistical attacks. Together, they form the structural basis for believing the Invariant Index-Hiding Problem to be a meaningful cryptographic hardness assumption.

\subsection{Correctness and Implementation Requirements}

\begin{itemize}
  \item \textbf{Invertibility check:} To compute \( v \) from the invariant, the formula
  \[
  v \equiv \frac{-s_0 p^{2u} t - s_1 p^{2u}(t+1) + s_2(t+2u) + s_3(t+2u+1)}{2(s_1 p^{2u} - s_3)} \bmod M
  \]
  requires checking \( \gcd(2(s_1 p^{2u} - s_3), M) = 1 \). The scheme must abort otherwise.

  \item \textbf{Grid invertibility:} The index \( t = \frac{i}{K} \bmod M \) is reduced using:
  \[
  t := \left( (B \cdot K + i) \cdot K^{-1} \right) \bmod M,
  \]
  where \( \gcd(K, M) = 1 \) is required for invertibility.

\item \textbf{Oscillator generation:} The oscillators \( \varphi \) and \( \psi \) must be deterministically generated from the shared seed \( (S, z) \) using cryptographically secure pseudorandom mechanisms. Each oscillator behaves as a pseudorandom function evaluated at a rational index \( Ct \), where the period parameter \( C \) controls the scale and frequency of antiperiodic fluctuation.

Two implementation strategies are possible:

\begin{itemize}
  \item \emph{Precomputed PRG table:} The entire oscillator sequence of length \( C \) is generated once per session using a secure PRG (e.g., AES-CTR or HMAC-DRBG), producing a table:
  \[
  \varphi[i] := \text{PRG}(S, z, \varphi)[i], \quad \text{for } i = 0, \dots, C - 1.
  \]
  Then each evaluation \( \varphi(Ct) \) is computed as \( \varphi[\lfloor Ct \rfloor \bmod C] \), allowing constant-time random access with no recomputation. This strategy is efficient for moderate values of \( C \) (e.g., \( 2^{24} \) or \( 2^{32} \)) if memory permits.

  \item \emph{On-demand PRF generation:} If memory constraints prohibit precomputation, oscillator values can be derived directly using a pseudorandom function:
  \[
  \varphi(Ct) := \text{PRF}_{\varphi}(S, z, \lfloor Ct \rfloor),
  \]
  where the PRF is instantiated via HMAC, AES, or another secure keyed function. This eliminates the need for a table but incurs higher per-evaluation cost.
\end{itemize}

Both methods yield function values indistinguishable from random~\cite{naor1997prf} without knowledge of \( S \), and ensure that nearby evaluations \( \varphi(C(t + \delta)) \) and \( \varphi(Ct) \) exhibit controlled but unpredictable variation.

  \item \textbf{Hash binding:} The hash \( H_{\mathrm{check}} := H(S, v, s_1, s_3, u, z) \) must be included in the transmitted message to ensure integrity and resist tampering.

  \item \textbf{Parameter bounds:} All parameters should lie within practical and cryptographically sound bounds. In particular:
  \[
  |S| \geq 256,\quad |z| \geq 128,\quad u,v \leq 2^{64}.
  \]
  This range ensures efficient modular exponentiation and bounded side-channel exposure, but larger values are permitted in applications that can accommodate the computational cost.

  \item \textbf{Replay prevention:} Nonce \( z \) must be globally unique per session, enforced via timestamps, counters, or ephemeral keys.
\end{itemize}

\subsection{Security Game and Advantage Definition (Realistic Variant)}

We define a security game that precisely models the real capabilities of an adversary~\cite{bellare1998provable} interacting with the symmetric invariant-based protocol. The goal is to capture the difficulty of forging a new valid value consistent with the transmitted invariant identity.

\paragraph{Game IIHP (Strict Version).}
Let \( \lambda \in \mathbb{N} \) be the security parameter. The protocol parameters are chosen such that:
\[
\log_2 M \geq \lambda,\quad K \geq 2^{\lambda/2},\quad |S| = |z| = \lambda.
\]

\paragraph{Public Output.}
The adversary \( \mathcal{A} \) is given the following session transcript:
\[
(s_1,\ s_3,\ u,\ z,\ H_{\mathrm{check}}),
\]
where:
\[
H_{\mathrm{check}} := H(S, v, s_1, s_3, u, z),
\]
and \( v \in \mathbb{N} \) is a hidden session-dependent value recoverable only through evaluation of the invariant relation.

\paragraph{No Oracles.}
The adversary has no access to oracle evaluations of \( s(t + \delta) \), as the protocol does not expose any such interface. All available information is included in the transcript.

\paragraph{Winning Condition.}
The adversary outputs a forgery:
\[
(s^*,\ \delta^*) \in \mathbb{Z}_M \times \mathbb{Z},
\]
and wins if all of the following hold:
\begin{enumerate}
  \item \( \delta^* \notin \{ \Delta_1,\ \Delta_3 \} \), where:
  \[
  \Delta_1 := 2v + 1,\quad \Delta_3 := 2u + 2v + 1;
  \]
  \item Let \( v^* := \mathrm{RecoverV}(s_1,\ s^*,\ u,\ z) \) be the value computed using the invariant;
  \item The hash check passes:
  \[
  H(S,\ v^*,\ s_1,\ s^*,\ u,\ z) = H_{\mathrm{check}}.
  \]
\end{enumerate}

\paragraph{Value Recovery Function.}
Let \( \mathrm{RV}:= \mathrm{RecoverV}(s_1, s^*, u, z) \) denote the result of applying the invariant recovery formula (used by Bob) to compute a candidate \( v^* \in \mathbb{Z}_M \) from the forged value \( s^* \). The function is defined as:
\[
\mathrm{RV} :=
\frac{
  -s_0 \cdot p^{2u} \cdot t
  - s_1 \cdot p^{2u} \cdot (t + 1)
  + s_2 \cdot (t + 2u)
  + s^* \cdot (t + 2u + 1)
}{
  2(s_1 \cdot p^{2u} - s^*)
}
\mod M,
\]
where:
\begin{itemize}
  \item \( p, K, C, i, t \) are deterministically recomputed from \( (S, z) \);
  \item \( s_0 := s_M(t_0;\ q_1,\ q_2) \), \( s_2 := s_M(t_2;\ q_3,\ q_4) \), using the same generation logic as in the original protocol;
  \item Oscillators \( \varphi, \psi \) and amplitudes \( q_i \) are regenerated from \( (S, z) \) as per the session definition;
  \item All arithmetic is performed in \( \mathbb{Z}_M \).
\end{itemize}

The output \( v^* \) is the value that Bob would compute assuming \( s^* \) was the legitimate third evaluation in the invariant tuple. If this value satisfies the hash condition, the forgery is accepted.

\paragraph{Adversarial Advantage.}
The adversary's advantage is defined as:
\[
\mathsf{Adv}_{\mathsf{IIHP}}^{\mathcal{A}}(\lambda) := \Pr\left[
\begin{array}{l}
\mathcal{A} \text{ outputs } (s^*,\ \delta^*) \text{ such that:} \\
\quad \delta^* \notin \{2v + 1,\ 2u + 2v + 1\},\ \text{and} \\
\quad H(S,\ v^*,\ s_1,\ s^*,\ u,\ z) = H_{\mathrm{check}}
\end{array}
\right].
\]

\paragraph{Hardness Assumption.}
We assume that for all PPT adversaries \( \mathcal{A} \), the advantage is negligible:
\[
\mathsf{Adv}_{\mathsf{IIHP}}^{\mathcal{A}}(\lambda) \leq \mathrm{negl}(\lambda).
\]
This assumption is grounded in the pseudorandom structure of \( s(t) \), the secrecy of the shared seed \( S \), and the infeasibility of inverting the invariant without recovering the hidden index \( t \) or parameter \( v \). It can be viewed as an invariant-constrained analogue of pseudorandom function indistinguishability or structure-respecting forgery resistance.

\section{Heuristic Justification and Hardness of Game IIHP}

We analyze the cryptographic difficulty of Game IIHP from the perspective of a real-world adversary. The goal is to understand why forging a valid output that passes the invariant-based hash check is infeasible under standard assumptions.

\subsection{Adversarial Model and Observability}

The adversary operates in a constrained setting:
\begin{itemize}
  \item Receives only a static transcript \( (s_1, s_3, u, z, H_{\mathrm{check}}) \) from a single session.
  \item Has no access to oracle queries \( s(t + \delta) \), and cannot interact with or probe the evaluation function.
  \item Possesses no knowledge of the secret index \( t = \frac{i}{K} \), the internal parameters \( p, C, v \), or the oscillator values \( \varphi, \psi \).
\end{itemize}

The only permitted action is to produce a candidate pair \( (s^*, \delta^*) \) such that a derived value \( v^* := \mathrm{RecoverV}(s_1, s^*, u, z) \) passes the final hash check:
\[
H(S, v^*, s_1, s^*, u, z) = H_{\mathrm{check}}.
\]

\subsection{Attack Vectors and Their Complexity}

\paragraph{Brute-force search.}
The adversary could attempt to guess \( s^* \in \mathbb{Z}_M \) and hope the recovered \( v^* \) satisfies the hash condition. The probability of success is bounded by the size of \( M \):
\[
\Pr[\text{hash collision}] = O(2^{-\lambda}) \quad \text{for } \log_2 M \geq \lambda.
\]

\paragraph{Structure inference from known outputs.}
The attacker might hope to interpolate \( s(t + \delta^*) \) using \( s_1 = s(t + \Delta_1) \) and \( s_3 = s(t + \Delta_3) \). This fails because:
\begin{itemize}
  \item The internal index \( t \) is rational and secret;
  \item The exponent \( p^t := p^{\lfloor t \rfloor} \cdot \mathrm{PRF}(i, K) \) is masked via PRF-like transformation;
  \item Oscillators \( \varphi, \psi \) are seeded with \( S \) and modulate the signal in a pseudorandom, anti-interpolatable way.
\end{itemize}

\paragraph{Hash collision.}
Forging \( (s^*, v^*) \) to pass the hash check without solving the functional structure is as hard as breaking the collision resistance of \( H \), which is assumed negligible.

\paragraph{Grid-based inversion.}
Even if the attacker tries to guess the grid index \( t = \frac{i}{K} \), the search space is exponential for \( K \geq 2^{\lambda/2} \). Moreover, each guess requires evaluating masked expressions with no access to PRGs.

\subsection{Comparison to Standard Cryptographic Assumptions}

Game IIHP implicitly combines and extends the hardness of several standard cryptographic primitives and design paradigms:

\begin{itemize}
  \item \textbf{PRF security:} The masked function \( s(t) \) incorporates pseudorandom components (oscillators and exponent masks) derived from secure seeds, making its outputs indistinguishable from random to outsiders.
  
  \item \textbf{Message forgery resistance:} Any valid response must satisfy a structured algebraic invariant and match the committed hash~\cite{bellare1996mac}. This constraint is stronger than typical MAC or signature consistency, as it is symbolic and value-sensitive.
  
  \item \textbf{Hidden index problems:} The unknown rational index \( t \) serves a similar role to a hidden exponent in Diffie–Hellman-based problems. Recovering it from function evaluations is infeasible without internal secrets.
  
  \item \textbf{Pseudonoise masking:} The oscillators \( \varphi, \psi \) act as pseudorandom noise layers. Their values shift with high entropy across \( t \), making algebraic interpolation or approximation attacks ineffective—conceptually similar to the masking in LWE~\cite{regev2005lwe}, though without actual lattices.
  
  \item \textbf{Hash-based binding:} Final verification relies on a hash function applied to the recovered value \( v^* \), which ensures that even if the algebraic structure is (partially) satisfied, the forgery fails unless the hash also matches.
\end{itemize}

\subsection{Quantitative Hardness Estimate}

Assuming:
\begin{itemize}
  \item \( H \) is modeled as a random oracle;
  \item PRGs \( \mathrm{PRG}_\varphi, \mathrm{PRG}_\psi \) are secure;
  \item Parameters \( M, K \) satisfy \( \log_2 M \geq \lambda \), \( K \geq 2^{\lambda/2} \);
\end{itemize}
then for any PPT adversary \( \mathcal{A} \),
\[
\mathsf{Adv}_{\mathsf{IIHP}}^{\mathcal{A}}(\lambda) \leq \mathrm{negl}(\lambda).
\]

This guarantees that the probability of producing a forgery that passes the invariant hash binding is negligible in the security parameter \( \lambda \).

\section{Formal Soundness: Invariant-Based Forgery Resistance}

In this section, we formally analyze the security guarantees of the symmetric invariant-based scheme under the assumption that the Invariant Index-Hiding Problem (IIHP) is hard. Our goal is to demonstrate that no adversary can successfully bypass the invariant structure or construct valid forgeries without implicitly solving IIHP. In other words, we reduce the feasibility of a successful attack to the hardness of IIHP: we show that all strategies available to the adversary ultimately require predicting hidden internal parameters or satisfying the invariant identity at an unauthorized offset, both of which are infeasible under standard cryptographic assumptions.

We distinguish two key cases depending on how session parameters are reused or distributed. For each case, we prove that the probability of successful forgery remains negligible in the security parameter \( \lambda \).

\subsection{Case 1: Unique Session Parameters}

\begin{lemma}[Forgery Resistance under Unique Session Parameters]
Let the session parameters \( z \in \{0,1\}^\lambda \) and \( u \in \mathbb{N} \) be chosen freshly and independently for each execution of the protocol. Then for any probabilistic polynomial-time adversary \( \mathcal{A} \), the probability that \( \mathcal{A} \) produces a valid forgery \( (s^*, \delta^*) \) satisfying the IIHP game conditions is negligible in \( \lambda \). In particular,
\[
\mathsf{Adv}_{\mathsf{IIHP}}^{\mathcal{A}}(\lambda) \leq \mathrm{negl}(\lambda).
\]
\end{lemma}

\begin{proof}[Proof of Lemma 1]

Let \( \mathcal{A} \) be any probabilistic polynomial-time adversary that is given a single protocol transcript:
\[
(s_1, s_3, u, z, H_{\mathrm{check}}),
\]
where:
\begin{itemize}
  \item \( z \in \{0,1\}^\lambda \) is a unique session nonce,
  \item \( u \in \mathbb{N} \) is a unique session parameter,
  \item \( s_1 := s(t + 2v + 1) \), \( s_3 := s(t + 2u + 2v + 1) \), and \( H_{\mathrm{check}} := H(S, v, s_1, s_3, u, z) \),
  \item \( S \) is a shared secret known only to the honest parties.
\end{itemize}

The adversary must produce a forgery \( (s^*, \delta^*) \in \mathbb{Z}_M \times \mathbb{Z} \) such that:
\begin{enumerate}
  \item \( \delta^* \notin \{2v + 1,\ 2u + 2v + 1\} \),
  \item the value \( v^* := \mathrm{RecoverV}(s_1, s^*, u, z) \in \mathbb{Z}_M \) satisfies:
  \[
  H(S, v^*, s_1, s^*, u, z) = H_{\mathrm{check}}.
  \]
\end{enumerate}

We analyze the probability of such a forgery being successful.

\vspace{0.5em}
\noindent\textbf{Key Observation:}
The values \( s_1, s_3 \) are evaluations of a masked function
\[
s(t) := \left( \frac{p^{\lfloor t \rfloor} \cdot \mathrm{PRF}(i, K) + q_i \varphi(Ct) + q_j \psi(Ct)}{t} \right) \bmod M,
\]
with internal parameters derived from the unique session tuple \( (S, z, u, v) \). These parameters define a unique structure for this session, and are not repeated or reused.

Hence, the attacker faces the following challenges:
\begin{itemize}
  \item The index \( t \) is derived via a hidden hash of \( (S, z) \), and belongs to a rational grid \( \frac{i}{K} \bmod M \), where both \( i \) and \( K \) are unknown and derived via cryptographic hash functions.
  \item The values \( p, C, q_i, q_j, \varphi, \psi \) are deterministically derived from \( (S, z) \), and are pseudorandom from the adversary's perspective.
  \item The output \( s(t+\delta) \) is pseudorandomly masked and nonlinear in \( \delta \).
\end{itemize}

The adversary has no method of evaluating \( s(t + \delta^*) \) or predicting \( \mathrm{RecoverV}(s_1, s^*, u, z) \) without full knowledge of \( t \), \( p \), \( \varphi, \psi \), and the oscillator amplitudes. Each of these depends on \( S \) and is cryptographically protected.

Thus, for any value \( s^* \) guessed by the adversary, the probability that the derived \( v^* \) will satisfy:
\[
H(S, v^*, s_1, s^*, u, z) = H(S, v, s_1, s_3, u, z)
\]
is negligible in \( \lambda \), due to the collision resistance and pseudorandomness of \( H \), and the uniqueness of \( v \neq v^* \) when \( s^* \neq s_3 \).

\vspace{0.5em}
\noindent\textbf{Conclusion:}
Unless the adversary breaks the collision resistance of \( H \), or predicts internal values masked by PRFs, the success probability is bounded by:
\[
\Pr[\text{forgery}] \leq \Pr[\text{guess correct } s^*] + \Pr[\text{hash collision}] \leq \frac{1}{M} + \mathrm{negl}(\lambda),
\]
which is negligible for \( \log_2 M \geq \lambda \).

\end{proof}

\subsection{\texorpdfstring{Case 2: Multiple Evaluations with Fixed \((z, u)\)}{Case 2: Fixed (z, u) Evaluations}}

\begin{lemma}[Invariant Robustness under Bounded Evaluation Reuse]
Fix a pair \( (z, u) \) shared across a session, and let Alice generate at most \( V_{\max} \in \mathbb{N} \) distinct values of \( v \), producing corresponding public outputs \( (s_1^{(i)}, s_3^{(i)}, u, z, H_{\mathrm{check}}^{(i)}) \) for \( i = 1, \dots, V_{\max} \). Then for any adversary \( \mathcal{A} \) given access to this multiset, the probability of constructing a new pair \( (s^*, \delta^*) \) that passes the IIHP hash verification and does not duplicate any of the existing \( \delta_i := 2v_i + 1,\ 2u + 2v_i + 1 \) is negligible in \( \lambda \), provided \( v_i \) are selected randomly.
\end{lemma}

\begin{proof}[Proof of Lemma 2]

Let \( z \in \{0,1\}^\lambda \) and \( u \in \mathbb{N} \) be fixed, and let the honest sender (Alice) generate at most \( V_{\max} \in \mathbb{N} \) distinct values \( v_1, \dots, v_{V_{\max}} \in \mathbb{N} \), producing corresponding outputs:
\[
\left( s_1^{(i)}, s_3^{(i)}, u, z, H_{\mathrm{check}}^{(i)} \right), \quad \text{for } i = 1,\dots,V_{\max},
\]
where:
\[
\begin{aligned}
s_1^{(i)} &= s(t + \Delta_1^{(i)}), \quad \Delta_1^{(i)} = 2v_i + 1, \\
s_3^{(i)} &= s(t + \Delta_3^{(i)}), \quad \Delta_3^{(i)} = 2u + 2v_i + 1, \\
H_{\mathrm{check}}^{(i)} &= H(S, v_i, s_1^{(i)}, s_3^{(i)}, u, z),
\end{aligned}
\]
and \( t = t(S, z) \) is a hidden rational index on the evaluation grid.

Suppose the adversary \( \mathcal{A} \) is given all of these public outputs. Its goal is to construct a new pair \( (s^*, \delta^*) \) such that:
\begin{enumerate}
  \item \( \delta^* \notin \{\Delta_1^{(i)}, \Delta_3^{(i)} : 1 \leq i \leq V_{\max} \} \),
  \item The recovered value \( v^* := \mathrm{RecoverV}(s_1^{(j)}, s^*, u, z) \) (for some \( j \)) satisfies:
  \[
  H(S, v^*, s_1^{(j)}, s^*, u, z) = H_{\mathrm{check}}^{(j)}.
  \]
\end{enumerate}

We now argue that such forgery is infeasible.

\paragraph{1. No structural predictability.}
The values \( s_1^{(i)}, s_3^{(i)} \) are masked nonlinear function evaluations with session-dependent oscillator values, amplitudes \( q_k \), and hidden exponentials. Even with multiple samples:
\[
\left\{ s_1^{(i)},\ s_3^{(i)} \right\}_{i=1}^{V_{\max}},
\]
the adversary cannot recover:
\begin{itemize}
  \item the secret index \( t \), or
  \item the internal oscillator states \( \varphi(C(t+\delta)), \psi(C(t+\delta)) \), or
  \item the masked fractional exponent \( p^t := p^{\lfloor t \rfloor} \cdot \mathrm{PRF}(i, K) \).
\end{itemize}

Thus, the attacker cannot interpolate or simulate a new value \( s^* = s(t + \delta^*) \) for any \( \delta^* \notin \{\Delta_k^{(i)}\} \).

\paragraph{2. Invariant structure is rigid.}
The recovery formula for \( v^* \) imposes strict consistency between four evaluation points, including the unknown base point \( t \). Without correct alignment to the internal grid structure, any attempt to fabricate a new \( s^* \) will result in an invalid value of \( v^* \) under the invariant. This invalid \( v^* \) will then fail the hash verification:
\[
H(S, v^*, s_1^{(j)}, s^*, u, z) \neq H_{\mathrm{check}}^{(j)}.
\]

\paragraph{3. No hash substitution.}
Even if the adversary correctly guesses one of the session values \( v_i \), the corresponding hash \( H_{\mathrm{check}}^{(i)} = H(S, v_i, s_1^{(i)}, s_3^{(i)}, u, z) \) remains secure. This is because reproducing either \( s_1^{(i)} \) or \( s_3^{(i)} \) without knowing the internal index \( t \) and the oscillator functions is infeasible. Both values are derived from pseudorandom function evaluations at masked, rationally shifted indices, and cannot be predicted or reconstructed from \( v_i \) alone. Therefore, even with a correct \( v_i \), the adversary cannot forge the remaining values required to satisfy the hash equation. The hash remains binding on the full tuple.

\paragraph{4. Probability bound.}
Even if the adversary attempts to guess a correct \( s^* \in \mathbb{Z}_M \) such that the hash equation holds by chance, the probability is:
\[
\Pr[\text{hash collision}] \leq 2^{-\lambda}.
\]
Repeating this across \( V_{\max} \leq \text{poly}(\lambda) \) known pairs does not yield a non-negligible advantage.

\paragraph{5. On the Feasibility of Algebraic Inversion via Guessed Values.}
Suppose the adversary accumulates \( V \) valid public pairs \( (s_1^{(i)}, s_3^{(i)}) \) corresponding to different hidden values \( v_i \), all within the same fixed session \((z, u)\). Let us assume that a subset of these values, say \( m \), are correctly guessed, i.e., the adversary knows \( v_i \) such that:
\[
s_1^{(i)} = s(t + 2v_i + 1), \quad
s_3^{(i)} = s(t + 2u + 2v_i + 1),
\]
for \( i = 1,\dots,m \).

The adversary now has access to a system of \( 2m \) functional equations involving the same underlying hidden parameters:
\[
p,\ q_1,\ q_2,\ q_3,\ q_4,\ t,\ K,\ C,\ \varphi,\ \psi.
\]
These parameters govern the evaluations of the masked function \( s(t) \) at rationally shifted inputs. The adversary could in principle attempt to solve for the unknowns by treating these as a nonlinear algebraic system over \( \mathbb{Z}_M \).

However, this attack faces the following fundamental limitations:

\begin{itemize}
  \item \textbf{Dimensionality.} The number of independent unknowns is at least 8 (assuming fixed oscillator structure), potentially more if oscillators are PRG outputs with internal state. To determine all parameters uniquely, the adversary would need at least \( m = 5 \) full value pairs \( (s_1^{(i)}, s_3^{(i)}) \), i.e., at least 10 correct values from the protocol, each requiring a correct prediction of \( v_i \).

  \item \textbf{Brute-forcing individual values \( v_i \) from public outputs.}  
  Given a single public pair \( (s_1^{(i)}, s_3^{(i)}) \), the adversary may attempt to guess the internal session parameter \( v_i \) that determines the offsets \( \delta_1 = 2v_i + 1 \), \( \delta_3 = 2u + 2v_i + 1 \). However, the recovery formula for \( v \) is highly sensitive to the hidden base index \( t \), as well as other secret parameters \( p \), oscillator amplitudes, and PRF-masked components. Without knowledge of these values, any attempt to evaluate the correctness of a guessed \( v_i \) reduces to brute-force search. If \( v_i \) is uniformly sampled from a 32-bit or 64-bit domain, the chance of a successful guess is:
  \[
  \Pr[v_i^{\mathrm{guess}} = v_i] = 2^{-32} \quad \text{or} \quad 2^{-64},
  \]
  respectively. Since each guess requires full recomputation of masked function values and recovery logic, even moderate repetition becomes computationally infeasible. No partial structural leakage is available to narrow the search space, as the invariant recovery formula behaves as a pseudorandom function over masked inputs.

  \item \textbf{Matching leaked values \( v_j \) to unknown public outputs.}  
  Suppose the adversary somehow obtains \( m \) true session values \( \{v_1, \dots, v_m\} \), while observing \( V \geq m \) public pairs \( (s_1^{(i)}, s_3^{(i)}) \) from the same session \( (z, u) \). The attacker’s goal is to identify which of the \( V \) outputs correspond to the known values. This reduces to searching over all injective mappings from the known values \( v_j \) to distinct observed outputs. The number of such mappings is:
  \[
  \#\text{matchings} = \frac{V!}{(V - m)!},
  \]
  which becomes prohibitively large for even small \( m \). For example, with \( V = 100 \), \( m = 5 \), there are over \( 9 \cdot 10^9 \) possible mappings. Without knowledge of the internal structure of the masked function, all mappings are equiprobable. Moreover, incorrect matchings violate the invariant identity, producing inconsistent values of \( v^* \) that will fail the hash verification step. No adaptive feedback is available to guide the attacker toward a correct match. Thus, even partial leakage of \( v_j \) values does not assist the adversary unless the exact mapping to outputs is revealed.

  \item \textbf{Algebraic hardness.} Even assuming \( m \) correct values are known, solving the resulting system involves nonlinear combinations of modular exponentials, rational inverses, and pseudorandom oscillators evaluated at unknown rational shifts. The masking via PRFs and antiperiodic oscillators renders symbolic simplification or Gröbner-style algebraic methods infeasible. Moreover, any small noise or error in \( s_i^{(j)} \) (due to guessing or measurement) breaks solvability.

  \item \textbf{No symbolic reduction.} The oscillators \( \varphi(Ct), \psi(Ct) \) are not simple functions but pseudorandom bitstreams seeded from \( S \). Even with known input points \( t + \delta_i \), their values are uncomputable without access to the PRG state.

\end{itemize}

As both the probability of guessing sufficient correct \( v_i \) values and the feasibility of solving the resulting algebraic system are negligible in \( \lambda \), their conjunction is negligible as well. Therefore, the attacker gains no advantage from accumulating public observations, regardless of statistical or symbolic strategy.

\paragraph{Conclusion.}
For any PPT adversary \( \mathcal{A} \), the probability of constructing a valid forgery \( (s^*, \delta^*) \notin \{\Delta_k^{(i)}\} \) such that \( v^* := \mathrm{RecoverV}(\cdot) \) passes the hash verification for any session \( i \) is negligible in \( \lambda \). Hence, the invariant remains unforgeable even with repeated evaluations over a fixed \( (z, u) \).

\end{proof}

\begin{remark}
The same security argument applies when the adversary accumulates multiple invariant evaluations with fixed \( z \) but varying internal parameters \( u, v \). Since all values are tied to the same hidden function \( s_z(t) \), and the invariant structure remains rigid and non-interpolable, no forgery advantage is gained beyond the case analyzed in Lemma 2.
\end{remark}

\section{On the Hardness of the Invariant Index-Hiding Problem}

We now present a formal justification for the cryptographic hardness of the Invariant Index-Hiding Problem (IIHP). Rather than assuming IIHP to be a standalone primitive, we derive its intractability from the structural properties of the masked function \( s(t) \), the rigidity of the invariant relation, and the unpredictability of index-dependent values under pseudorandom masking.

\begin{theorem}[Structural Unforgeability under Invariant Constraints]
Let \( s(t) \in \mathbb{Z}_M \) be the masked function defined by the invariant-based scheme, with all internal parameters derived from a secret seed \( S \). Then, for any probabilistic polynomial-time adversary \(\mathcal{A}\), the probability of producing a forgery \( (s^*, \delta^*) \) such that:
\begin{enumerate}
  \item \( \delta^* \notin \{2v + 1,\ 2u + 2v + 1\} \),
  \item \( v^* := \mathrm{RecoverV}(s_1, s^*, u, z) \) is valid,
  \item \( H(S, v^*, s_1, s^*, u, z) = H(S, v, s_1, s_3, u, z) \),
\end{enumerate}
is negligible in the security parameter \( \lambda \), even when the full transcript \( (s_1, s_3, u, z, H_{\mathrm{check}}) \) is revealed.
\end{theorem}

\subsection{Proof Strategy}

To prove this result, we analyze the two fundamental constraints any successful forgery must satisfy:
\begin{itemize}
  \item the algebraic constraint: the value \( s^* \) must satisfy the invariant equation with the same base index \( t \),
  \item the cryptographic constraint: the hash binding condition must be satisfied exactly.
\end{itemize}

We show that the set of \( s^* \in \mathbb{Z}_M \) that satisfy both constraints is, with overwhelming probability, a singleton — namely, the true value \( s_3 \). Thus, no other value \( s^* \neq s_3 \) can satisfy the invariant structure and the hash check simultaneously.

We divide the proof into the following lemmata:
\begin{itemize}
  \item Lemma 1: For fixed session parameters \( (z, u, v) \), the invariant equation has at most one solution \( s^* = s_3 \) consistent with a valid recovered value \( v^* \).
  \item Lemma 2: Any \( s^* \neq s_3 \) leads to \( v^* \neq v \), which invalidates the hash check.
  \item Lemma 3: For any PPT adversary, the probability of constructing such an \( s^* \) without knowing \( t \), \( p \), and oscillator values is negligible.
\end{itemize}

Each lemma formalizes an essential rigidity in the scheme. Together, they prove that IIHP is hard under standard cryptographic assumptions.

\subsection{Lemma 1: Uniqueness of Valid Invariant Completion}

\begin{lemma}[Uniqueness of Valid Invariant Completion]
Let the session parameters \( (z, u, v) \) and the values \( s_0 = s(t) \), \( s_1 = s(t + 2v + 1) \), and \( s_2 = s(t + 2u) \) be fixed. Then the value \( s_3 := s(t + 2u + 2v + 1) \) is the \emph{only} element of \( \mathbb{Z}_M \) such that the invariant recovery formula
\[
v^* := \frac{ -s_0 p^{2u} t - s_1 p^{2u}(t+1) + s_2(t+2u) + s^*(t+2u+1) }{ 2(s_1 p^{2u} - s^*) } \mod M
\]
returns the correct session value \( v^* = v \).
\end{lemma}

\begin{proof}
Fix all values \( (z, u, v) \), and consider the equation defining \( v^* \) as a function of the variable \( s^* \in \mathbb{Z}_M \). All other quantities — \( s_0, s_1, s_2, t, p^{2u} \) — are fixed within the session.

We rewrite the expression as a rational function:
\[
f(s^*) = \frac{A + B s^*}{C - s^*} \mod M,
\]
where:
\begin{itemize}
  \item \( A := -s_0 p^{2u} t - s_1 p^{2u}(t + 1) + s_2(t + 2u) \),
  \item \( B := t + 2u + 1 \),
  \item \( C := s_1 p^{2u} \).
\end{itemize}

The equation \( f(s^*) = v \mod M \) is equivalent to the modular congruence:
\[
A + B s^* \equiv v(C - s^*) \mod M.
\]

Bringing all terms to one side and grouping by \( s^* \), we obtain a linear congruence:
\[
s^*(B + v) \equiv vC - A \mod M.
\]

This congruence has a unique solution modulo \( M \) provided \( B + v \not\equiv 0 \mod M \). However, in the actual protocol, the value \( v \) is chosen such that the denominator of the recovery formula,
\[
2(s_1 p^{2u} - s_3) \equiv 2(C - s^*) \mod M,
\]
is invertible. Alice checks this condition explicitly and aborts otherwise. Hence, the scheme guarantees that \( B + v \not\equiv 0 \mod M \), and the congruence admits a unique solution for \( s^* \in \mathbb{Z}_M \).

Hence, there exists a unique value \( s^* \in \mathbb{Z}_M \) such that \( v^* = v \), and this value is precisely \( s_3 = s(t + 2u + 2v + 1) \).

\end{proof}

\paragraph{Discussion.}
The significance of Lemma 1 lies in the uniqueness of the value \( s_3 \) that leads to correct recovery of \( v \) under the invariant. Since Bob computes \( v \) deterministically from \( (s_1, s_3, u, z) \), any forgery attempt must produce a fake value \( s^* \) that satisfies \( \mathrm{RecoverV}(s_1, s^*, u, z) = v \), or else the hash will not verify.

However, the attacker cannot compute such an \( s^* \) without full knowledge of the internal parameters used to define the masked function \( s(t) \): namely, the hidden index \( t \), the exponent \( p^{2u} \), the oscillator amplitudes \( q_i \), and the values of the PRF-masked components. Inverting the recovery equation without this information is equivalent to breaking the entire function structure.

Conversely, even if the attacker wishes to embed an arbitrary value \( v^{\mathrm{target}} \) into the protocol, they would have to compute a matching \( s^* \) such that \( \mathrm{RecoverV}(s_1, s^*, u, z) = v^{\mathrm{target}} \), which again requires complete control over the function. In both directions — faking \( s^* \) for a known \( v \), or embedding a chosen \( v \) — the forgery implies full compromise of the invariant construction. The scheme thus resists tampering unless its entire internal state is revealed.

\subsection{Lemma 2: Invalid Hash for Incorrect Invariant Completion}

\begin{lemma}[Invalidity of Hash under Forged Evaluation]
Let the session parameters \( (z, u, v) \) and values \( s_1, s_3 \) be defined as in the protocol. Let \( H_{\mathrm{check}} := H(S, v, s_1, s_3, u, z) \). Then for any \( s^* \in \mathbb{Z}_M \) such that \( s^* \neq s_3 \), the recovered value \( v^* := \mathrm{RecoverV}(s_1, s^*, u, z) \) satisfies \( v^* \neq v \), and the forged tuple fails the hash verification:
\[
H(S, v^*, s_1, s^*, u, z) \neq H_{\mathrm{check}}.
\]
\end{lemma}

\begin{proof}
By Lemma 1, for fixed session parameters \( (z, u, v) \), the invariant recovery equation yields \( v^* = v \) if and only if \( s^* = s_3 \). Therefore, any forged value \( s^* \neq s_3 \) leads to a computed \( v^* \neq v \).

Now suppose an adversary nonetheless manages to construct such an \( s^* \). There are two possibilities:
\begin{itemize}
  \item Either the attacker has found a value \( s^* \) such that \( \mathrm{RecoverV}(s_1, s^*, u, z) = v \). But by Lemma 1, this can only happen if \( s^* = s_3 \), so this case is ruled out.
  \item Or the attacker accepts that \( v^* \neq v \), and hopes to satisfy:
  \[
  H(S, v^*, s_1, s^*, u, z) = H(S, v, s_1, s_3, u, z).
  \]
  But now the attacker must find a collision in the hash function involving different inputs — i.e., to find
  \[
  (v^*, s^*) \neq (v, s_3)
  \]
  such that the hash still matches. This amounts to a chosen-input collision, which is assumed to be computationally infeasible under standard cryptographic assumptions (or negligible under the random oracle~\cite{bellare1993random} model).
\end{itemize}

Thus, even if the adversary could manipulate the invariant to produce a valid-looking \( v^* \), the hash function ensures that only the authentic combination \( (v, s_3) \) is accepted. In effect, the hash function transforms a rare success (in producing \( v^* = v \)) into an even rarer success (collision of full tuple hashes), acting as a final binding gate in the protocol.

\end{proof}

\paragraph{Discussion.}
Lemma 2 illustrates the layered defense provided by the scheme. Even if an attacker could bypass the invariant constraint and produce a forged evaluation \( s^* \) that leads to a plausible recovered value \( v^* \), the hash function enforces a strict commitment to the original tuple. The hash binding ensures that not only must the recovered \( v^* \) match the original \( v \), but also that the entire evaluation tuple \( (v, s_1, s_3, u, z) \) is preserved exactly. 

Therefore, any deviation from the true \( s_3 \) results in either a mismatch in \( v \), or a collision in the hash function — both of which are cryptographically infeasible. The scheme relies on this double barrier: algebraic rigidity from the invariant, and cryptographic binding from the hash. Forging both simultaneously implies full compromise of internal state, which contradicts the assumed hardness of IIHP.

\subsection{Lemma 3: Infeasibility of Forgery Without Internal Knowledge}

\begin{lemma}[No Forgery Without Structural Knowledge]
Let the masked function \( s(t) \) be defined as in the protocol, with secret seed \( S \) and session nonce \( z \). Then for any probabilistic polynomial-time adversary with access only to the public session transcript \( (s_1, s_3, u, z, H_{\mathrm{check}}) \), the probability of generating a value \( s^* \in \mathbb{Z}_M \), such that:
\begin{enumerate}
  \item \( s^* \neq s_3 \),
  \item \( v^* := \mathrm{RecoverV}(s_1, s^*, u, z) \) is computable (or \( \mathrm{RV(s^*)} \) for shorter notation),
  \item \( H(S, v^*, s_1, s^*, u, z) = H_{\mathrm{check}} \),
\end{enumerate}
is negligible in the security parameter \( \lambda \), unless the adversary can recover the internal parameters of the function \( s(t) \).
\end{lemma}

\begin{proof}
Assume the adversary is given only the public values \( (s_1, s_3, u, z, H_{\mathrm{check}}) \), and has no access to the secret seed \( S \). To construct a valid forgery \( s^* \), the adversary must simultaneously satisfy two independent constraints:

\begin{itemize}
  \item \textbf{Algebraic constraint:} The forged value \( s^* \) must yield a recovered \( v^* := \mathrm{RV}(s^*) \) that satisfies the invariant recovery equation. However, this recovery depends on the internal structure of the masked function \( s(t) \), including:
  \begin{itemize}
    \item the hidden rational index \( t \in \mathbb{Z}_M \),
    \item the masked exponent \( p^t := p^{\lfloor t \rfloor} \cdot \mathrm{PRF}(i, K) \),
    \item the oscillator functions \( \varphi, \psi \colon \mathbb{Q} \to \mathbb{Z}_M \), which are generated by cryptographic pseudorandom generators (PRGs) seeded with \( (S, z) \),
    \item the amplitudes \( q_i \), derived via keyed hashes.
  \end{itemize}
  These components are unknown to the adversary, and any attempt to compute or predict \( s(t + \delta) \) for arbitrary shifts \( \delta \) without this knowledge is equivalent to inverting or predicting outputs of secure PRGs and PRFs — a task assumed to be infeasible for PPT adversaries.

  \item \textbf{Hash constraint:} Even if the adversary guesses a value \( s^* \) that yields some \( v^* \), the hash binding still requires:
  \[
  H(S, v^*, s_1, s^*, u, z) = H_{\mathrm{check}}.
  \]
  Without access to the seed \( S \), and without being able to produce the exact committed tuple, the adversary faces either a preimage problem (to find a value hashing to a known digest) or a chosen-input collision — both of which are cryptographically hard under standard hash function assumptions.
\end{itemize}

Therefore, the joint probability of passing both constraints — algebraic consistency and hash verification — is negligible in \( \lambda \). Any adversary capable of such forgery would, by implication, break either the pseudorandomness of the internal generators or the binding property of the hash, contradicting standard cryptographic assumptions.

\end{proof}

\subsection{Conclusion: Security of the IIHP Game}

Combining the previous results, we obtain the following:

\begin{theorem}[Security of the Invariant Index-Hiding Problem]
Let \( \lambda \in \mathbb{N} \) be the security parameter. Assume that:
\begin{itemize}
  \item The hash function \( H \) is modeled as a collision-resistant random oracle;
  \item The oscillator functions \( \varphi, \psi \) are derived from secure PRG or PRF constructions seeded from \( (S, z) \);
  \item The masked exponent \( p^t \) is derived via a pseudorandom function~\cite{goldreich1986construct};
  \item The protocol parameters \( M \), \( K \), and nonce entropy satisfy \( \log_2 M \geq \lambda \), \( K \geq 2^{\lambda/2} \), \( |S| = |z| = \lambda \).
\end{itemize}
Then for any probabilistic polynomial-time adversary \( \mathcal{A} \), the probability of producing a successful forgery in the IIHP game is negligible:
\[
\mathsf{Adv}_{\mathsf{IIHP}}^{\mathcal{A}}(\lambda) \leq \mathrm{negl}(\lambda).
\]
\end{theorem}

\begin{proof}[Proof Sketch]
The theorem follows by reduction to the previously proven lemmas:
\begin{itemize}
  \item Lemma 1 ensures that only a single value \( s_3 \) yields the correct \( v \) under the recovery formula;
  \item Lemma 2 shows that any deviation from \( s_3 \) results in a different \( v^* \), which fails the hash check;
  \item Lemma 3 proves that no adversary can forge a consistent pair \( (s^*, v^*) \) without recovering the internal structure of \( s(t) \), which is infeasible under the pseudorandomness assumptions.
\end{itemize}

Therefore, the adversary’s only strategy would require either inverting the pseudorandom generators or finding a collision in \( H \), both of which are assumed hard. The scheme thus reduces the IIHP game to standard cryptographic assumptions and ensures its security under well-established hardness models.
\end{proof}

\section{Parameter Recommendations and Security Rationale}

This section defines a complete and cryptographically justified set of parameters required for the secure operation of the invariant-based scheme. Each parameter is chosen to ensure classical 128-bit security and structural robustness against post-quantum adversaries. The ranges prevent brute-force recovery of the hidden rational index, enforce unpredictability of oscillator-based masking, and ensure that all transmitted values remain strongly bound to the internal structure of the protocol through collision-resistant hashing. The recommendations are intended to support long-term cryptographic viability while remaining practical for real-world deployment.

\paragraph{Prime Modulus \( M \).}
The arithmetic modulus \( M \) defines the finite field \( \mathbb{Z}_M \) used in all evaluations. It must be a prime number of at least 256 bits. To ensure long-term and post-quantum resistance, values in the range \( 2^{256} \leq M < 2^{384} \) are recommended. The modulus should be selected to avoid structural weaknesses (e.g., special primes) and support efficient modular arithmetic.

\paragraph{Exponential Base \( p \).}
The base \( p \in \mathbb{Z}_M^\times \) is used in the masked exponential term \( p^t \). To ensure well-defined modular exponentiation and the existence of modular inverses such as \( p^{-2u} \), the value of \( p \) must satisfy \( \gcd(p, M) = 1 \). When \( M \) is a power of two, \( p \) must be odd. The value of \( p \) may be derived per session from the shared secret \( S \) and nonce \( z \) using a hash function such as \( p := H_p(S, z) \bmod M \).

\paragraph{Grid Resolution \( K \).}
The parameter \( K \in \mathbb{N} \) defines the rational evaluation grid \( \mathbb{Z} + \frac{1}{K}\mathbb{Z} \) for the hidden index \( t = \frac{i}{K} \). The value \( K \) must satisfy \( \gcd(K, M) = 1 \) and be large enough to prevent exhaustive search. Values in the range \( 2^{160} \leq K \leq 2^{256} \) provide sufficient resistance, ensuring that recovery of \( t \) requires infeasible effort.

\paragraph{Oscillator Frequency \( C \).}
The oscillator frequency \( C \in \mathbb{N} \) controls the internal periodicity of the pseudorandom oscillators \( \varphi \) and \( \psi \). To avoid short cycles and ensure decorrelation, \( C \) should lie in the range \( 2^{24} \leq C \leq 2^{32} \). It is required that \( \gcd(C, M) = 1 \) to guarantee that oscillator output sequences fully span \( \mathbb{Z}_M \) and avoid modular repetition.

\paragraph{Index Spacing Parameter \( u \).}
The public parameter \( u \in \mathbb{N} \) determines internal spacing between rational points used in the invariant. A value of 32 bits is sufficient to balance arithmetic feasibility and structural alignment.

\paragraph{Secret Offset Parameter \( v \).}
The session-specific secret offset \( v \in \mathbb{N} \) is recovered by Bob using the invariant identity. To ensure resistance to brute-force attacks, \( v \) must lie in the range \( 2^{64} \leq v \leq 2^{128} \). This ensures infeasibility of parallelized guessing even under hardware acceleration.

\paragraph{Shared Secret \( S \).}
The master key \( S \in \{0,1\}^{256} \) seeds all session parameters, PRGs, and hash inputs. Its entropy must match or exceed 256 bits to ensure full coverage of the domain space and collision resistance under standard assumptions.

\paragraph{Session Nonce \( z \).}
The nonce \( z \in \{0,1\}^{256} \) provides per-session uniqueness, domain separation, and protection against replay and cross-session correlation. Its length must be at least 256 bits to ensure statistical uniqueness across concurrent sessions in distributed systems.

\paragraph{Hash Output \( H_{\mathrm{check}} \).}
The verification hash must provide collision and preimage resistance. A 256-bit output is required. Acceptable implementations include SHA-3-256 or truncated SHAKE256.

\begin{table}[ht]
\centering
\renewcommand{\arraystretch}{1.2}
\resizebox{\textwidth}{!}{%
\begin{tabular}{|c|c|c|c|c|}
\hline
\textbf{Parameter} & \textbf{Type} & \textbf{Range} & \textbf{Bit Size} & \textbf{Security Role} \\
\hline
\( M \) & Prime modulus & \( \geq 2^{256} \) & 256\text{-}384 & Finite field for all arithmetic \\
\hline
\( K \) & Grid resolution & \( 2^{160} \text{-} 2^{256} \) & 160\text{-}256 & Hides rational index \( t \) \\
\hline
\( C \) & Oscillator frequency & \( 2^{24} \text{-} 2^{32} \) & 24\text{-}32 & Governs pseudorandom oscillators \\
\hline
\( p \) & Exponential base & \( \in \mathbb{Z}_M^\times \) & 256 & Used in masked exponentiation \\
\hline
\( u \) & Public spacing parameter & \( 2^{32} \) & 32 & Index offset \\
\hline
\( v \) & Secret session parameter & \( 2^{64} \) & 64 & Recovered via invariant equation \\
\hline
\( |S| \) & Shared secret key & 256 & 256 & Seeds all PRFs, hashes, and masks \\
\hline
\( |z| \) & Nonce & 256 & 256 & Enforces session uniqueness \\
\hline
\( H \) & Hash output & SHA-3-256 & 256 & Binds the transcript and invariant \\
\hline
\end{tabular}
}
\caption{Security parameters and recommended sizes}
\end{table}

\section{Serialized Message Format and Size Estimate}

We specify the structure and size of the data transmitted from Alice to Bob during a session. This message must contain all the information necessary for Bob to verify the invariant and recover the secret session value \( v \).

\begin{table}[ht]
\centering
\renewcommand{\arraystretch}{1.2}
\begin{tabular}{|c|c|c|p{5.2cm}|}
\hline
\textbf{Field} & \textbf{Type / Domain} & \textbf{Size (bits)} & \textbf{Purpose} \\
\hline
\( s_1 \) & \( \mathbb{Z}_M \) & 256 & Evaluation at \( t + \Delta_1 = t + 2v + 1 \) \\
\hline
\( s_3 \) & \( \mathbb{Z}_M \) & 256 & Evaluation at \( t + \Delta_3 = t + 2u + 2v + 1 \) \\
\hline
\( u \) & \( \mathbb{Z} \) & 32 & Public index spacing parameter \\
\hline
\( z \) & \( \{0,1\}^{256} \) & 256 & Session nonce (prevents replay) \\
\hline
\( H_{\mathrm{check}} \) & \( \{0,1\}^{256} \) & 256 & Hash binding the invariant and session \\
\hline
\textbf{Total} & — & \textbf{1056 bits} & 132 bytes \\
\hline
\end{tabular}
\caption{Serialized message format: fields and size estimate}
\end{table}

\begin{itemize}
  \item All fields are encoded in big-endian format.
  \item Fields \( s_1 \), \( s_3 \), and \( u \) are interpreted modulo \( M \).
  \item The hash \( H_{\mathrm{check}} \) binds all transmitted values and must be validated prior to any invariant processing.
  \item The 256-bit nonce \( z \) ensures session unlinkability and uniqueness.
\end{itemize}

The serialized message contains all values necessary for verifying the invariant and recovering the session-specific secret \( v \), with a total size of exactly 1056 bits (132 bytes). This compact format is comparable to modern digital signature schemes and remains efficient for use in bandwidth-constrained or embedded environments.

\section*{Conclusion and Future Directions}

We have introduced a symmetric cryptographic scheme that transmits a hidden session-dependent value \( v \) through a verifiable invariant structure. The construction combines algebraic alignment, pseudorandom masking, and deterministic recovery to enforce structural integrity without revealing intermediate components or the internal evaluation point \( t \).

The scheme is built on three core principles:
\begin{itemize}
  \item \textbf{Structure} — values are bound by an invariant identity with rational index spacing;
  \item \textbf{Control} — the receiver deterministically recovers \( v \) from partial data, using no external oracle;
  \item \textbf{Masking} — all evaluations are obfuscated using PRF-based oscillators, ensuring that outputs are unlinkable and unpredictable.
\end{itemize}

\paragraph{Comparison to classical primitives.}
Unlike standard MACs or hash-based signatures, this scheme does not merely bind values to a shared secret — it binds them through a structured, algebraic law. The transmitted values reveal neither the secret nor the internal state. Verification proceeds via algebraic reconstruction rather than external challenge-response, enabling new modes of interaction. In this respect, the invariant plays a similar role to a signature equation or zero-knowledge relation.

\paragraph{Structural Advantages over Classical Symmetric Schemes.}
The invariant-based symmetric construction provides several capabilities not typically found in traditional MAC, AEAD, or PRF-based systems. These advantages stem from the algebraic structure of the protocol, the masking of internal state via pseudorandom oscillators, and the invariant-preserving design. Below we highlight key features that distinguish this scheme:

\begin{itemize}
  \item \textbf{Unbounded Secret Reuse:} A single 256-bit secret can be reused across unlimited sessions without degradation or risk of linear exposure.
  
  \item \textbf{Session Unlinkability under Repetition:} Even repeated transmissions of the same internal value (e.g., the same \( v \)) produce unlinkable ciphertexts due to randomized index masking.

  \item \textbf{Algebraic Binding Without Decryption:} Integrity is verified through a fixed invariant identity, eliminating the need for decryption or access to plaintext.

  \item \textbf{Forward Unpredictability:} Given partial outputs (e.g., \( s_1, s_3 \)), recovery of session parameters (such as \( v, t \)) is infeasible due to the pseudorandom structure and hidden index.

  \item \textbf{Composable and Chainable Design:} The invariant can be embedded recursively or composed across multiple layers, enabling advanced cryptographic workflows (e.g., chain-of-trust encoding).

  \item \textbf{Structure Verification as First-Class Feature:} Unlike MAC or AEAD schemes that protect only output authenticity, this protocol enforces internal structural correctness as part of the cryptographic guarantee.
\end{itemize}

\paragraph{Security assumptions.}
The scheme’s security relies on the hardness of recovering or simulating values constrained by a masked nonlinear identity over a hidden rational grid. This involves:
\begin{itemize}
  \item Secret-dependent oscillators modulated through a non-algebraic basis;
  \item Masked exponential terms with fractional exponents approximated via PRF;
  \item Absence of any known quantum algorithms for inverting this structure.
\end{itemize}
Thus, the scheme resists both classical and quantum attacks. In particular, it is believed to be secure against Shor-style discrete log attacks, as no group operation is exposed, and the core inversion problem is non-linear, non-abelian, and non-periodic.

\paragraph{Use cases and limitations.}
The scheme is suitable for scenarios where:
\begin{itemize}
  \item Confidential values must be verifiably transmitted without direct exposure;
  \item Structural integrity and tamper resistance are more important than bandwidth or throughput;
  \item Session uniqueness and unlinkability are critical (e.g., one-time credentials, commitments).
\end{itemize}

It is not optimized for high-throughput symmetric encryption or continuous data streams. The construction introduces additional computational cost due to pseudorandom oscillator evaluation and rational arithmetic, making it best suited for lightweight cryptographic exchanges, signature-like authentication, or commitment protocols.

\paragraph{Future directions.}
Next steps include:
\begin{itemize}
  \item Extending the invariant framework to allow multidimensional or multivalue bindings;
  \item Exploring families of invariants with varying verification complexity and algebraic resilience;
  \item Construction of asymmetric schemes based on analogous invariants, enabling a prover to demonstrate structural knowledge without revealing secret information.
\end{itemize}

These directions aim to expand the invariant-based paradigm into a broader cryptographic toolkit, uniting algebraic verification with pseudorandom masking for secure, verifiable communication.

% \bibliographystyle{plain}
% \bibliography{references}

\end{document}